\documentclass[runningheads]{llncs}
\usepackage[T1]{fontenc}
\usepackage{graphicx}
\usepackage[graphicx]{realboxes}
\usepackage{csquotes}
\usepackage{orcidlink}
\usepackage{amsmath, amsfonts, amssymb}
\usepackage{hyperref}
\usepackage{dsfont}
\usepackage{marginnote}
\usepackage{xcolor}
\usepackage{orcidlink}
\usepackage{cancel}
\usepackage{mathptmx}
\usepackage{stmaryrd}
 \usepackage{multirow} 
 \usepackage{lipsum}
\usepackage{subfig}
\usepackage{booktabs}
\usepackage{acronym}
\usepackage{cancel}
\usepackage{float}
\usepackage{mathtools}
\usepackage{xcolor}
\usepackage{longtable} 
\usepackage{lipsum}
\usepackage{scrextend}
\usepackage{pdflscape}
\usepackage{rotating}
\usepackage{tikz}
\usepackage{graphicx}
\usepackage{makecell}
\usetikzlibrary{decorations.pathmorphing,calc,shapes,shapes.geometric,patterns}
\usetikzlibrary{arrows,positioning,calc,intersections}
\usetikzlibrary{datavisualization.formats.functions}
\usetikzlibrary{shapes.multipart,matrix}
\definecolor{tu0}{rgb}{0.7451, 0.1176, 0.2353}

\definecolor{tu1}{rgb}{1.0000, 0.8039, 0.0000}
\definecolor{tu11}{rgb}{1.0000, 0.8627, 0.3020}
\definecolor{tu12}{rgb}{1.0000, 0.9020, 0.4980}
\definecolor{tu13}{rgb}{1.0000, 0.9412, 0.6980}
\definecolor{tu14}{rgb}{1.0000, 0.9608, 0.8000}

\definecolor{tu2}{rgb}{0.9804, 0.4314, 0.0000}
\definecolor{tu21}{rgb}{0.9882, 0.6039, 0.3020}
\definecolor{tu22}{rgb}{0.9882, 0.7137, 0.4980}
\definecolor{tu23}{rgb}{0.9922, 0.8275, 0.6980}
\definecolor{tu24}{rgb}{0.9961, 0.8863, 0.8000}

\definecolor{tu3}{rgb}{0.6902, 0.0000, 0.2745}
\definecolor{tu31}{rgb}{0.7529, 0.2000, 0.4196}
\definecolor{tu32}{rgb}{0.8431, 0.4980, 0.6353}
\definecolor{tu33}{rgb}{0.9216, 0.7490, 0.8196}
\definecolor{tu34}{rgb}{0.9529, 0.8510, 0.8902}

\definecolor{tu4}{rgb}{0.4863, 0.8039, 0.9020}
\definecolor{tu41}{rgb}{0.6431, 0.8627, 0.9333}
\definecolor{tu42}{rgb}{0.7412, 0.9020, 0.9490}
\definecolor{tu43}{rgb}{0.8431, 0.9412, 0.9686}
\definecolor{tu44}{rgb}{0.8980, 0.9608, 0.9804}

\definecolor{tu5}{rgb}{0.0000, 0.5020, 0.7059}
\definecolor{tu51}{rgb}{0.3020, 0.6510, 0.7961}
\definecolor{tu52}{rgb}{0.5490, 0.7765, 0.8667}
\definecolor{tu53}{rgb}{0.7490, 0.8745, 0.9255}
\definecolor{tu54}{rgb}{0.8510, 0.9255, 0.9569}

\definecolor{tu6}{rgb}{0.0000, 0.3255, 0.4549}
\definecolor{tu61}{rgb}{0.2510, 0.4941, 0.5922}
\definecolor{tu62}{rgb}{0.5490, 0.6941, 0.7529}
\definecolor{tu63}{rgb}{0.7490, 0.8314, 0.8627}
\definecolor{tu64}{rgb}{0.8510, 0.8980, 0.9176}

\definecolor{tu7}{rgb}{0.7765, 0.9333, 0.0000}
\definecolor{tu71}{rgb}{0.8431, 0.9529, 0.3020}
\definecolor{tu72}{rgb}{0.8863, 0.9647, 0.4980}
\definecolor{tu73}{rgb}{0.9333, 0.9804, 0.6980}
\definecolor{tu74}{rgb}{0.9569, 0.9882, 0.8000}

\definecolor{tu8}{rgb}{0.5373, 0.6431, 0.0000}
\definecolor{tu81}{rgb}{0.6784, 0.7490, 0.3020}
\definecolor{tu82}{rgb}{0.7686, 0.8196, 0.4980}
\definecolor{tu83}{rgb}{0.8588, 0.8941, 0.6980}
\definecolor{tu84}{rgb}{0.9059, 0.9294, 0.8000}

\definecolor{tu9}{rgb}{0.0000, 0.4431, 0.3373}
\definecolor{tu91}{rgb}{0.3020, 0.6118, 0.5373}
\definecolor{tu92}{rgb}{0.5490, 0.7490, 0.7020}
\definecolor{tu93}{rgb}{0.7490, 0.8588, 0.8353}
\definecolor{tu94}{rgb}{0.8549, 0.9176, 0.9059}

\definecolor{tu10}{rgb}{0.8000, 0.0000, 0.6000}
\definecolor{tu101}{rgb}{0.8706, 0.3490, 0.7412}
\definecolor{tu102}{rgb}{0.9216, 0.6000, 0.8392}
\definecolor{tu103}{rgb}{0.9608, 0.8000, 0.9216}
\definecolor{tu104}{rgb}{0.9804, 0.8980, 0.9608}

\definecolor{tu110}{rgb}{0.4627, 0.0000, 0.4627}
\definecolor{tu111}{rgb}{0.5961, 0.2510, 0.5961}
\definecolor{tu112}{rgb}{0.7294, 0.4980, 0.7294}
\definecolor{tu113}{rgb}{0.8392, 0.6980, 0.8392}
\definecolor{tu114}{rgb}{0.9216, 0.8510, 0.9216}

\definecolor{tu120}{rgb}{0.4627, 0.0000, 0.3294}
\definecolor{tu121}{rgb}{0.6118, 0.3020, 0.5333}
\definecolor{tu122}{rgb}{0.7569, 0.5490, 0.6980}
\definecolor{tu123}{rgb}{0.8667, 0.7490, 0.8314}
\definecolor{tu124}{rgb}{0.9216, 0.8510, 0.9020}

\definecolor{tu130}{rgb}{0.0314, 0.0314, 0.0314}
\definecolor{tu131}{rgb}{0.3725, 0.3725, 0.3725}
\definecolor{tu132}{rgb}{0.5882, 0.5882, 0.5882}
\definecolor{tu133}{rgb}{0.7529, 0.7529, 0.7529}
\definecolor{tu134}{rgb}{0.8667, 0.8667, 0.8667}

\definecolor{tu140}{rgb}{0.0000, 0.6875, 0.3125}

\newcommand{\mmchange}[1]{\textcolor{black}{#1}}

\newcommand{\cdchange}[1]{\textcolor{black}{#1}}

\newcommand\mcD{\ensuremath{\mathcal{D}}}

\newcommand\mcP{\ensuremath{\mathcal{P}}}

\newcommand\mcX{\ensuremath{\mathcal{X}}}
\newcommand\mcY{\ensuremath{\mathcal{Y}}}

\newcommand\nats{I\!\! N}

\newcommand\treq{\triangleq}
\newcommand\fall{\mbox{ for all }}

\newcommand{\watanabeset}{\mathcal{S}_Q(\gamma)}
\newcommand{\hayashiset}{\mathcal{T}_P(\gamma)}
\newcommand{\D}{{{\mathcal D}}}
\newcommand{\PP}{{{\mathcal P}}}
\newcommand{\X}{{{\mathcal X}}}
\newcommand{\Y}{{{\mathcal Y}}}
\newcommand{\setS}{{{\mathcal S}}}
\begin{document}
\title{Converse Techniques for Identification via Channels}
%
%
\author{Larissa~Brüche\inst{1} \and
Marcel~A.~Mross\inst{1,4}\orcidlink{0000-0003-1747-6876} \and
Yaning~Zhao\inst{1,3}\orcidlink{0009-0003-3601-5503} \and Wafa~Labidi\inst{1,2,3}\orcidlink{0000-0001-5704-1725} \and
Christian~Deppe\inst{1,3}\orcidlink{0000-0002-2265-4887} \and Eduard~A.~Jorswieck\inst{1,4}\orcidlink{0000-0001-7893-8435} }
\authorrunning{L. Brüche et al}
%
\institute{Technische Universität Braunschweig, Institute for Communications Technology, Braunschweig, Germany \and Technical University of Munich, TUM School of Computation, Information and Technology, Munich, Germany 
\and 6G-life, 6G research hub, Germany
\and 
6G-RIC Research and Innovation Cluster, Germany\\
\email{l.brueche@tu-bs.de, m.mross@tu-bs.de, yaning.zhao@tu-bs.de, wafa.labidi@tum.de, christian.deppe@tu-bs.de, e.jorswieck@tu-bs.de}}
\maketitle              
\begin{abstract}
\mmchange{The model of \textit{identification via channels}, introduced by Ahlswede and Dueck, has attracted increasing attention in recent years.} Unlike in Shannon's classical model, where the receiver aims to determine which message was sent from a set of $M$ messages, message identification focuses solely on discerning whether a specific message $m$ was transmitted. The encoder can operate deterministically or through randomization, with substantial advantages observed particularly in the latter approach. While Shannon's model allows transmission of $M = 2^{nC}$ messages, Ahlswede and Dueck's model facilitates the identification of $M = 2^{2^{nC}}$ messages, exhibiting a double exponential growth in block length.
In their seminal paper, Ahlswede and Dueck established the achievability and introduced a "soft" converse bound. Subsequent works have further refined this, culminating in a strong converse bound, applicable under specific conditions. Watanabe's contributions have notably enhanced the applicability of the converse bound.
The aim of this survey is multifaceted: to grasp the formalism and proof techniques outlined in the aforementioned works, analyze Watanabe’s converse, trace the evolution from earlier converses to Watanabe’s, emphasizing key similarities and differences that underpin the enhancements. 
Furthermore, we explore the converse proof for message identification with feedback, also pioneered by Ahlswede and Dueck. By elucidating how their approaches were inspired by preceding proofs, we provide a comprehensive overview. This overview paper seeks to offer readers insights into diverse converse techniques for message identification, with a focal point on the seminal works of Hayashi, Watanabe, and, in the context of feedback, Ahlswede and Dueck.

\end{abstract}
\section{Introduction}

\subsection{Historical Overview}
\label{sec:historical_development}

Inspired by the seminal works of Yao \cite{yao} and \mmchange{JáJá} \cite{jaja}, Ahlswede and Dueck delved into the realm of message identification through communication channels \cite{ahlswede}. \mmchange{They discovered that the number of identifiable messages grows double-exponentially, in strong contrast to the traditional exponential growth observed in classical communication.} This contribution earned them the prestigious Best Paper Award from the IEEE Information Theory Society.

\mmchange{However, the \textit{soft converse} provided in \cite{ahlswede} only holds under the condition that the error probability converges to zero exponentially.
Han and Verdú \cite{han_paper} developed the concept of \textit{channel resolvability} to prove a strong converse for the double exponential coding theorem.}
Steinberg \cite{steinberg} extended this idea and developed \textit{partial channel resolvability} to improve the strong converse and apply it to general channels. 
\mmchange{However, Hayashi \cite{hayashi} pointed out a gap in the proof of Lemma 2 in \cite{steinberg}. He derived non-asymptotic formulas for identification via channels and channel resolvability and applied these to the wiretap channel. In a related work,  Oohama \cite{oohama}, also studies the strong converse for identification via channels using similar methods.} Oohama used these results to prove that the sum of the two error probability tends to one exponentially as $n$ goes to infinity at transmission rates above the ID-capacity. \mmchange{Both works use the \textit{soft covering lemma} as the main tool to prove bounds for the channel resolvability. The concept of soft covering of a distribution by a codebook was introduced by Wyner \cite[Theorem 6.3]{wyner1975common}. He developed this tool to prove achievability in his study on the common information of two random variables. A later work by Hayashi and Watanabe \cite{hayashi_watanabe} investigates the strong converse of channel resolvability and uses the soft covering lemma to derive second-order rates under certain conditions. Table \ref{tab:overview_publications_intentions} gives an overview of some publications on channel resolvability and identification converses.}

\begin{table}[ht!]
    \centering
    \caption{Development of identification converses.}
    \label{tab:overview_publications_intentions}
    \begin{tabular}{|c|c|c|}
        \hline
        \textbf{Publication} & \textbf{Year} & \textbf{\mmchange{Contribution}} \\
        \hline
        \mmchange{Han \& Verdú \cite{han_verdu}} & \mmchange{1992} & \parbox{7.5cm}{\mmchange{establish the strong converse for identification using channel resollvability}} \\
        \hline
        Steinberg \cite{steinberg}&1998& \parbox{7.5cm}{\mmchange{prove strong converse for general channel, later shown to have a gap}} \\
        \hline
        Hayashi \cite{hayashi}& 2006 &\parbox{7.5cm}{\mmchange{derive non-asymptotic formulas based for channel resolvability and identification, apply it to the wiretap channel}} \\
        \hline
        Oohama \cite{oohama}& 2013&\parbox{7.5cm}{use soft covering lemma to prove that the sum of two error probabilities converge to one exponentially} \\
        \hline
        Hayashi \& Watanabe \cite{hayashi_watanabe}& 2014&\parbox{7.5cm}{\mmchange{use soft covering lemma to derive second-order rate expansion for channel resolvability}} \\
        \hline
        \mmchange{Watanabe \cite{watanabe}}& \mmchange{2022} &\parbox{7.5cm}{\mmchange{use partial resolvability to prove converse for general channel and second-order rate expansion.}} \\
        \hline
        \end{tabular}
\end{table}


Ahlswede's work in \cite{ahlswedeconcepts} presents a relatively short and conceptually straightforward proof of the converse theorem for identification via the discrete memoryless channel (DMC). In his approach, Ahlswede revisits the initial idea from \cite{ahlswede}, which involves replacing distributions with uniform distributions on "small" subsets, specifically those with cardinalities slightly above a certain threshold. The proof in \cite{ahlswedeconcepts} primarily relies on the theories of large deviations and hypergraphs. A detailed description of this proof, along with its generalization to quantum communication, is available in \cite{Ahlswedewinter}. Additionally, \cite{ahlswede_book} provides a comparison between this combinatorial method and the proof by Han and Verdú. This converse method has been applied to calculate capacities for the compound wiretap channel and the arbitrarily varying wiretap channel \cite{bochedepperobust}, \cite{bochedeppejamming}, as well as for quantum channels \cite{Ahlswedewinter}. In this survey, we will not delve into the combinatorial approach, as it is already comprehensively described in \cite{Ahlswedewinter} and \cite{ahlswede_book}.

It has been shown in \cite{Shannon56} that the capacity of a DMC is not increased by the availability of a feedback channel, even if the feedback channel is noiseless and has unlimited capacity. However, feedback can significantly reduce the complexity of encoding or decoding. A straightforward code construction for a DMC with feedback was explored in \cite{ConstructiveProofAhlswede}. Furthermore, \cite{MACtransmissionFeedback,DUECK19801,KramerFeedback} demonstrated that feedback enhances the capacities of discrete memoryless multiple-access channels and discrete memoryless broadcast channels. Additionally, \cite{Ahlswede2006} noted that noiseless feedback can be used to generate a secret key shared exclusively between the transmitter and the legitimate receiver.

Identification via arbitrarily varying channels (AVC) with noiseless feedback was investigated in \cite{Ahlswede2000}. Identification over discrete multi-way channels with complete feedback was presented in \cite{ID_Multiway_Feedback}. A unified theory of identification via channels with finite input and output alphabets in the presence of noisy feedback was established in \cite{GeneralTheory}. Additionally, the secure identification capacity over the discrete wiretap channel with secure feedback was studied in \cite{Ahlswede2006}.

In \cite{labidifeedback}, the Gaussian channel with feedback is considered. For a positive noise variance, a coding scheme is proposed that generates infinite common randomness between the sender and the receiver. It is shown that any rate for identification via the Gaussian channel with noiseless feedback can be achieved. The result holds regardless of the selected scaling for the rate. This result was generalized in \cite{wiesefeedback}
for general additive noise channels.

The first proof for identification with feedback was developed for both deterministic and stochastic encoding by Ahlswede and Dueck \cite{ahlswede_feedback}. Their proof bears similarity to Wolfowitz's earlier proof for transmission with feedback \cite[Theorem 4.8.2, p.95]{wolfowitz}. \mmchange{A thorough introduction into the theory of identification via channels as well as an extensive overview over results in the field can be found in \cite{ahlswede_book}.}

The structure of this survey will be as follows: After this introduction, we will present the main definitions and notations needed in Section \ref{sec:notation}. The problem formulations for identification via channels and identification with feedback, as well as the method of information spectrum quantities, will be presented in Sections \ref{problem_id},\ref{problem_idf} and \ref{sec:information_spectrum}. In Section \ref{id}, we will explore the converse technique for identification via channels based on channel resolvability. Additionally, we will analyze the converse proof of identification in the presence of noiseless feedback as presented by Ahlswede and Dueck in Section \ref{idf}. Our approach will be to first present Wolfowitz's method for transmission with feedback, followed by a comparative analysis with the findings of Ahlswede and Dueck.

\subsection{Definitions and Notations}
\label{sec:notation}
We revisit first some foundational results about message transmission through channels. 
\begin{definition}
A discrete memoryless channel (DMC) 
is a triple $(\X,\Y,W)$, where $\X$ is the finite input alphabet, 
$\Y$ is the finite output alphabet and   
\begin{equation} W=\left\{W(y|x):x\in\X, y\in\Y\right\}\end{equation} is a stochastic matrix.
The probability for a sequence $y^n\in\Y^n$ to be received if 
$x^n\in\X^n$ was sent
$$
W^n(y^n|x^n)=\prod_{j=1}^n W(y_j|x_j).
$$
\end{definition}
If the definitions of $\X$ and $\Y$ are clear, we simply refer to $W$ as the DMC.

In Shannon's model of information transmission, the sender's task is to encode messages as sequences of channel input symbols in such a way that, even if the channel does not transmit the sequence perfectly, the receiver can still correctly identify the sent message with high probability. For an arbitrary set $\setS$, we denote by $\mathcal{P}(\setS)$ the set of all \textit{probability distributions} on $\setS$. Let $P$ be a probability distribution on a finite set $\X$. The \textit{entropy} of $P$ is defined as
$
H(P):= -\sum_{x\in\X}P(x)\log P(x).
$
If $X$ is a random variable with distribution $P$, we denote the 
entropy of $X$ by
$
H(X):= H(P).
$
Let $X,Y$ be RVs on finite sets $\X,\Y$ with distributions
$P_X$ and $P_{Y|X}$. 
The \textit{conditional entropy} of $Y$ given $X$ is defined by
$$
H(Y|X):= -\sum_{x\in\X}P_X(x)\sum_{y\in\Y}P_{Y|X}(y|x)\log P_{Y|X}(y|x).
$$
Let $X$ and $Y$ be random variables on finite sets $\X$ and $\Y$ with distributions
$P_X$ and $P_{Y}$, 
respectively. Then we define the \textit{mutual information} between $X$ and $Y$ by
$$
I(X;Y):= H(Y)-H(Y|X).
$$
If it is clear which alphabets are
to be used, we omit them if we are
talking about the channel.
If $P$ is a probability distribution on $\X$ and 
$W=\left\{W(y|x):x\in\X, y\in\Y\right\}$, a stochastic matrix, we set
$
I(P,W):= I(X;Y),
$
where $X$ is a RV with distribution $P$ and $Y$ has conditional
distribution $W(\cdot|x)$, given $X=x$. 

\begin{definition}\label{Code}
A randomized $\left(n, M, \epsilon\right)$-transmission code for a channel $W$
is a family of pairs
$
\left\{ \left( Q_i,\D_i \right) | i=1, \ldots , M\right\} 
$
such that $\forall\  i=1,\cdots,M$ and $\forall\  i \not= j$
\begin{equation}
Q_i \in \PP\left( \X^n \right),
\D_i \subset \Y^n,\ 
\end{equation}
\begin{equation}
\label{randcctrdisjoint}
\D_i\cap\D_j = \emptyset,
\end{equation}
\begin{equation}\label{randcctrerrbound}
\sum_{x^n\in\X^n}{Q_i\left(x^n\right)W^n\left(\D_i|x^n\right)}
\ge 1-\epsilon.
\end{equation}
\end{definition}

Note that in classical transmission often deterministic encoding is used. This means we have no randomization at the input and therefore a probability of sending the message $i$ equal to one ($Q_i=1$).\

\begin{definition}
Let $W$ be a DMC.
\begin{enumerate}
\item The rate $R$ of a $(n, M,\epsilon)$ code is defined as
$R = \frac {\log M}n$ bits, i.e., $M=2^{nR}$.
\item A rate $R$ is said to be achievable if for all $\epsilon\in (0,1)$ there exists a $n_0(\epsilon)$, such 
that for all $n\geq n_0(\epsilon)$ there exists a $(n,2^{nR},\epsilon)$ code.
\item The transmission capacity $C(W)$ of a DMC $W$ is the supremum of all achievable rates.
\end{enumerate}
\end{definition}

Let $M(n,\epsilon)\treq\max\left\{M\in\nats:\textrm{A }(n,M,\epsilon)
\textrm{-Code exists}\right\}$. Then we have the following   

\begin{theorem}[Shannon's Coding Theorem]\label{thm:shannon}
Let $\epsilon\in(0,1)$ be fixed. Then
\begin{equation}\label{dmccapacity}
\lim_{n\to\infty}\frac{\log M(n,\epsilon)}{n}=
\max_{P\in\mcP(\mcX)}I(P,W)= C(W)
\end{equation}
\end{theorem}

$C$ depends only on
the matrix $W$. The error probability $\epsilon$ only affects the speed 
of convergence.

Note that the theorem holds regardless of whether we use a deterministic $(n, M, \epsilon)$ code or a randomized $(n, M, \epsilon)$ code.

Denote the marginal output distribution induced by a channel input distribution $P$ on a channel $W$ with 
\begin{align}
    PW(y) = \sum_{x\in \mathcal{X}}P(x)W(y|x) 
\end{align}
and the joint distribution of channel input and output with
\begin{align}
    P\times W (x, y) = P(x)W(y|x).
\end{align}
The joint distribution of two statistically independent random variables $X$ and $Y$ is also denoted by
\begin{align}
    P \times Q(y) = P(x)Q(y).
\end{align}
The (normalized) variational distance $d(\cdot , \cdot)$ is defined by
\begin{equation}
    d(P,Q) := \frac{1}{2} \sum\limits_{x \in \mathcal{X}} |P(x)-Q(x)|. \label{eq:variational-distance}
\end{equation}
We follow the convention in \cite{watanabe}, which includes the factor $\frac{1}{2}$ in the definition of the variational distance. Not that it is also very common to omit this factor, as e.g. in \cite{hayashi}. Using the definition \eqref{eq:variational-distance}, the variational distance is equal to the total variation distance:
\begin{equation}
    d(P,Q) = \sup_{\mathcal{A} \subseteq \mathcal{X}} |P(\mathcal{A}) - Q(\mathcal{A})|.\label{eq:total-variation-distance}
\end{equation}

We refer to a sequence $\mathrm{\mathbf{W}} = \{W^n\}_{n = 1}^{\infty}$ of conditional distributions $W^n(y^n|x^n)$, $(x^n, y^n) \in \mathcal{X}^n \times \mathcal{Y}^n$, as a \textit{general channel}.

\subsection{\mmchange{Background}}
\label{sec:information_spectrum}

This chapter will give a brief overview over some \mmchange{necessary concepts and tools. For a more detailed introduction, see  \cite[Chapter 2]{tan} and \cite{han}.}

Consider a binary hypothesis testing between a null hypothesis $Z \sim P_Z$, indicated by 0 and an alternative hypothesis $Z \sim Q_Z$, indicated by 1. The hypothesis test can be described by a channel $W: \mathcal{Z} \to \{0,1\}$. Then we will find the error probability of first kind
\begin{equation}
    P_I(W) = \sum\limits_z P_Z(z) W(1|z) 
\end{equation}
and of the second kind
\begin{equation}
    P_{II}(W) = \sum\limits_z Q_Z(z)W(0|z) .
\end{equation}
Table \ref{tab:variables_hypothesis_testing} shows an overview of the hypothesis testing and its variables. We define the optimal type II error probability under the condition that the type I error probability is less or equal than $\varepsilon$ by
\begin{equation}
    \label{eq:optimal_type_II_error_prob}
    \beta_{\varepsilon}(P_Z, Q_Z) = \inf_{W: P_I(W)\leq\varepsilon} P_{II}(W) .
\end{equation}

\begin{table}[h!]
    \centering
    \caption{Variables of the hypothesis testing}
    \label{tab:variables_hypothesis_testing}
    \begin{tabular}{|c|c|c|}
        \hline
        hypothesis & null hypothesis & alternative hypothesis \\
        \hline
        probability & $P_Z$ & $Q_Z$ \\
        indicator & 0 & 1 \\
        channel decision & $W(0|z)$ & $W(1|z)$ \\
        \hline
        error probability & \parbox{5cm}{1st kind:\newline $P_I(T) = \sum\limits_z P_Z(z) W(1|z)$} & \parbox{5cm}{2nd kind:\newline $P_{II}(T) = \sum\limits_z Q_Z(z)W(0|z)$} \\
        description & \parbox{5cm}{channel decides for alternative hypothesis although null hypothesis was send} & \parbox{5cm}{channel decides for null hypothesis although alternative hypothesis was send}\\
        \hline
        \end{tabular}
\end{table}

\begin{definition}
    The \emph{$\varepsilon$-information spectrum divergence} \cite[eq. (2.9)]{tan} is given by
\begin{equation}
    \label{eq:epsilon_spectral_inf_divergence}
    D_\mathrm{s}^{\varepsilon}(P_Z\|Q_Z) = \sup \left\{\gamma \in \mathbb{R}: \mathbb{P}\left(\log\frac{P_Z(z)}{Q_Z(z)}\leq\gamma\right)\leq\varepsilon \right\}.
\end{equation}
\end{definition}

In order to get a better understanding of this definition, Figure \ref{fig:e_spectral_divergence} shows a visualization on the probability distribution. We choose the supremum of all thresholds $\gamma$, such that the probability distribution of $\log\frac{P_Z(z)}{Q_Z(z)}$ is less than or equal to $\varepsilon$.
\begin{figure}[!ht]
     \centering
        \includegraphics[width=0.65\textwidth]{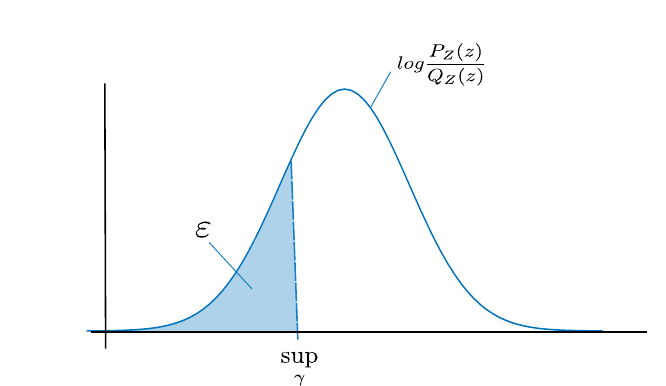}
     \caption{Visualization of the $\varepsilon$-spectral inf divergence}
     \label{fig:e_spectral_divergence}
\end{figure}

The following Lemma connects the $\varepsilon$-information spectrum divergence to the optimal type-II error of a binary hypothesis test.
\begin{lemma}[{\cite[Lemma 2.4]{tan}}]
    \label{lem:divergence_leq_logbeta}
    For $0\leq\varepsilon<1$, it holds that
    \begin{equation}
        D_\mathrm{s}^{\varepsilon}(P_Z\|Q_Z) \leq -\log\beta_{\varepsilon}(P_Z, Q_Z) \leq D_\mathrm{s}^{\varepsilon+\zeta}(P_Z\|Q_Z) + \log\left(\frac{1}{\zeta}\right) 
    \end{equation}
    for any $0<\zeta<1-\varepsilon$.
\end{lemma}

Now consider the hypothesis testing between 
\begin{equation}
    P \times W(x, y) \hspace{1cm} \text{and} \hspace{1cm} P \times Q(x, y), 
\end{equation}
where P describes an input and Q describes an output distribution.
An important quantity for the analysis of the asymptotics of channel coding is
\begin{align}
    \label{eq:inf-sup-beta}
    \inf\limits_{P \in \mathcal{P}(\mathcal{X})} \sup\limits_{Q\in \mathcal{P}(\mathcal{Y})} \beta_{\varepsilon}(P\times W, P\times Q) .
\end{align}
For finite alphabets $\mathcal{X}$ and $\mathcal{Y}$, the following saddle-point property holds.

\begin{lemma}[Saddle-Point Property, \cite{polyanskiy}]
    \label{lem:saddle_point_property}
    Let $0\leq\varepsilon<1$, then the optimal value in \ref{eq:inf-sup-beta} is attainable and
    \begin{equation}
        \min_{P \in \mathcal{P}(\mathcal{X})} \max_{Q\in \mathcal{P}(\mathcal{Y})} \beta_{\varepsilon}(PW, PQ) = \max_{Q\in \mathcal{P}(\mathcal{Y})} \min_{P \in \mathcal{P}(\mathcal{X})} \beta_{\varepsilon}(PW, PQ). 
    \end{equation}
\end{lemma}

Thus far, we have seen non-asymptotic quantities that are used to derive non-asymptotic bounds. To express asymptotic capacity results for general channels, we also need the following definition.
\begin{definition} Consider a general channel $\mathrm{\mathbf{W}} = \{W^n\}_{n = 1}^{\infty}$ with input $\mathrm{\mathbf{X}} = \{x^n\}_{n = 1}^{\infty}$ and output $\mathrm{\mathbf{Y}} = \{Y^n\}_{n = 1}^{\infty}$. The \emph{$\varepsilon$-spectral-inf mutual information rate} is defined by 
\begin{equation}
    \label{eq:e-spectral-inf-mutual-information}
    \underline{I}^{\varepsilon}(\mathrm{\mathbf{X}};\mathrm{\mathbf{Y}}) := \sup \left\{a: \limsup_{n\to\infty} \mathbb{P}\left(\frac{1}{n} \log\frac{W^n(Y^n|X^n)}{P_{Y^n}(Y^n)} \leq a \right) \leq \varepsilon \right\}.
\end{equation}
The \emph{spectral-inf mutual information rate} is defined by 
    \begin{equation}
    \label{eq:e-spectral-inf-mutual-information}
    \underline{I}(\mathrm{\mathbf{X}};\mathrm{\mathbf{Y}}) := \sup \left\{a: \lim_{n\to\infty} \mathbb{P}\left(\frac{1}{n} \log\frac{W^n(Y^n|X^n)}{P_{Y^n}(Y^n)} \leq a \right) = 0 \right\}
\end{equation}
and the \emph{spectral-sup mutual information rate} is defined as
\begin{equation}
    \label{eq:e-spectral-inf-mutual-information}
    \overline{I}(\mathrm{\mathbf{X}};\mathrm{\mathbf{Y}}) := \inf \left\{a: \lim_{n\to\infty} \mathbb{P}\left(\frac{1}{n} \log\frac{W^n(Y^n|X^n)}{P_{Y^n}(Y^n)} \geq a \right) = 0 \right\}.
\end{equation}
\end{definition}

\subsection{Problem Formulations of Identification via channels}
\label{problem_id}
In this section, we formulate the identification via channels problem, and state the related results. 
In the context of message identification, the receiver is tasked with determining whether a received message corresponds to a specific message $i$ that he possesses. The receiver must make this determination based on the information received through the communication channel, aiming for a high probability of correct identification.
\begin{definition}\label{defidcode}
A randomized $\left(n, N, \epsilon, \lambda\right)$ 
identification code (abbreviated as ID-code) is a family of pairs
$$
\left\{ \left( Q_i, 
        \mcD_i \right) | i=1, \ldots , N\right\} 
$$
with
$$
Q_i \in \mcP\left( \mcX^n \right),\quad 
\mcD_i \subset \mcY^n \fall i=1,\ldots,N
$$
and with errors of first resp.\ second kind bounded by
\begin{equation}\label{idcodeerrbound1}
\varepsilon_n^{(i)}=\sum_{x\in\mcX^n}{Q_i\left(x^n\right)W^n\left(\mcD_i|x^n\right)}
\ge 1-\epsilon \fall i=1, \ldots , N
\end{equation}
and 
\begin{equation}\label{idcodeerrbound2}
\delta_n^{(i,j)}=\sum_{x\in\mcX^n}{Q_j\left(x^n\right)W^n\left(\mcD_i|x^n\right)}
\le\lambda \fall i,j=1, \ldots , N, \quad i\ne j
\end{equation}
\end{definition}

The receiver who is interested in message $i$ will decide that his message
was transmitted iff the received channel output is in $\mathcal{D}_i$,
otherwise he will deny that message $i$ was sent.

The two types of errors $\epsilon$ and $\lambda$ will differ in their origins: errors of the first type $\epsilon$ are caused by channel noise, while errors of the second kind $\lambda$ primarily result from the identification (ID) code scheme.

The main difference compared to transmission codes is that the disjointness condition for decoding sets is replaced by the weaker property \eqref{idcodeerrbound2}. Instead of a single receiver interested in a specific message, one can imagine a scenario where all decoders are in the same location. Each receiver adds "his" message to a common list if he believes that his message has been sent. This suggests that ID-codes are somewhat similar to list-codes. While list-codes typically impose a limit on the list size, ID-codes make it unlikely for a message to be included in the decoding list unless it is the sent message. In both cases, there is a high probability that the sent message will be in the decoding list.
Then we define the optimal code size of identification via channel $W$ as follows.
\begin{align}
    N^*(\varepsilon, \lambda|W) := \sup\{N: \text{an } (N, \varepsilon, \lambda)-\text{ID-code exists for the channel } W\}.
\end{align}

\begin{theorem}[Ahlswede, Dueck, Han, Verd\'u]\label{thm:double_exponent_coding_theorem}
Let 
${\epsilon,\lambda\le\frac{1}{2}}$. Then
\begin{equation}
\lim_{n\to\infty}\frac{\log\log N^*(\varepsilon, \lambda|W^n)}{n}=C(W),
\end{equation}
where $C(W)$ again denotes the channel capacity from Theorem\ref{thm:shannon}.
\end{theorem}

The double exponent coding theorem is one of the major gains of Ahlswede's work \cite{ahlswede}. According to Shannon's theorem (see Theorem \ref{thm:shannon}), the number of messages $N$ grows exponentially with the number of bits $n$ ($\sim 2^{nC}$), with Theorem \ref{thm:double_exponent_coding_theorem} by Ahlswede $N$ grows double exponentially with $n$ ($\sim 2^{2^{nC}}$)\footnote{Strictly, basis 2 is only correct for binary codes. For codes with $b$ symbols, chose $b$ as a basis.}. This fact can be attributed due to the nature of identification and transmission codes. 
In a transmission code with a message set of size $M$, the receiver faces an $M$-hypothesis testing problem, requiring  $M$ pairwise disjoint decoding sets. In an identification code, the receiver faces, independently of the size of the message set, a two-hypothesis testing on not necessarily pairwise disjoint decoding sets, deciding between the hypothesis, whether the message was sent or not. The possible overlap of the decoding regions is proportional to the allowed error probabilities. In other words: we tolerate a certain error probability to accept not disjoint decoding regions.

Table \ref{tab:comparison_transmission_ID} shows the main characteristics of the ID code compared to the transmission code. The ID-code is applicable for boolean decisions at the receiver, while in transmission the content of the message is of interest. Identification is realized by not disjunct decoding sets, resulting in a second error type. The first kind error is defined similarly to the transmission error. As a result we obtain double exponentially growth for the coding rate in identification instead of exponential growth, as in transmission.
\begin{table}[ht!]
    \centering
    \caption{Comparison of transmission and ID-codes}
    \label{tab:comparison_transmission_ID}
    \begin{tabular}{|c|c|c|}
        \hline
        & \textbf{Transmission} & \textbf{Identification} \\
        \hline
        intention & \parbox{4cm}{\centering receiver wants to know which message is sent} & \parbox{4.5cm}{\centering receiver just wants to know whether message is sent or not} \\ 
        \hline
        decoding sets &\parbox[c]{4cm}{\centering\vspace{0.01cm}\includegraphics[width=3cm]{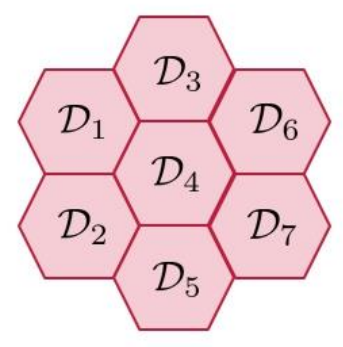}} & \parbox[c]{4.5cm}{\centering\vspace{0.01cm} \includegraphics[width=3cm]{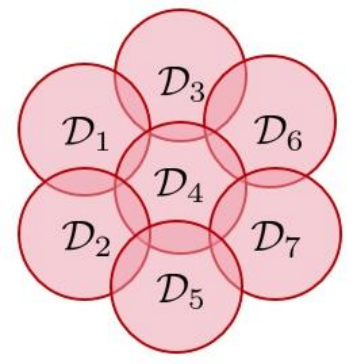}}\\
        \hline
        error probabilities & $W^n(\mathcal{D}_i^c|x_i) \leq \varepsilon$ & \makecell{first kind: $\sum_{x\in\mcX^n}Q_i\left(x^n\right)W^n\left(\mcD_i|x^n\right)$ \\ second kind: $\sum_{x\in\mcX^n}Q_j\left(x^n\right)W^n\left(\mcD_i|x^n\right)$ }\\
        \hline
        number of messages & $\sim 2^{nC}$ & $\sim 2^{2^{nC}}$ \\
        \hline
        \end{tabular}
\end{table}
\newline
The direct part of Theorem \ref{thm:double_exponent_coding_theorem} was proven by Ahlswede and Dueck in 1989, while the converse was only proven under the assumption that the error probabilities converge exponentially to zero \cite{ahlswede}. \mmchange{The authors called this a \textit{soft converse:} 
\begin{align}
    \limsup_{n\to\infty} \frac{1}{n} \log\log N^*(2^{-n\epsilon}, 2^{-n\delta}|W^n) \leq C \qquad \forall \epsilon, \delta > 0.
\end{align}
As already outlined in \cite{ahlswede}, a \textit{weak converse} in this notation means
\begin{align}
    \inf_{\epsilon, \delta \in (0,1)} \limsup_{n\to\infty} \frac{1}{n} \log\log N^*(\epsilon, \delta|W^n) \leq C 
\end{align}
and a strong converse means
\begin{align}
    \limsup_{n\to\infty} \frac{1}{n} \log\log N^*(\epsilon, \delta|W^n) \leq C \qquad \forall \epsilon, \delta > 0.
\end{align}
The strong converse was first proved by Han and Verdú \cite{han_verdu}. In Section \ref{id}, we present two proofs of the strong converse by Hayashi \cite{hayashi} and Watanabe \cite{watanabe}.
}

\mmchange{The identification coding theorem not only holds for DMCs, but was proven for the class of \textit{general channels}, which are neither required to be stationary nor ergodic. We limit ourselves to discrete input and output alphabets $\mathcal{X}$ and $\mathcal{Y}$. A \textit{general channel} is a sequence $\mathbf{W} = \{W^n\}_{n=1}^{\infty}$, where $W^n$ are arbitrary conditional probability distributions, interpreted as a channels from $\mathcal{X}^n$ and $\mathcal{Y}^n$.}
We first define the following terms for the identification problem via a general channel $\mathbf{W}$.
\begin{definition}
\label{def:e_d_achievable}
    Given $\varepsilon\geq 0 $, $\delta<1$ a rate $R$ of an ID-code is $(\varepsilon,\delta)$-achievable if there exists an $(N_n,\varepsilon_n,\delta_n)$-ID-code satisfying
    \begin{align}
        &\limsup_{n \to \infty} \varepsilon_n \leq \varepsilon, \\
        &\limsup_{n \to \infty} \delta_n \leq \delta, \\
        &\liminf_{n \to \infty} \frac{1}{n} \log\log N_n \geq R,
    \end{align}
    where $\varepsilon_n$ and $\delta_n$ is the maximal type I and type II error probability, respectively, i.e.,
    \begin{align}
        \varepsilon_n&=\max_{1\le i\le N_n}\varepsilon_n^{(i)},\\
        \delta_n&= \max_{1\le i\ne j\le N_n}\delta_n^{(i,j)}.
    \end{align}
\end{definition}

\begin{definition}[{\cite[Definition 6.2.2]{han}}]
\mmchange{The ($\epsilon, \delta$)-ID capacity for a general channel $\mathbf{W}$ is defined as}
\begin{align} 
    C_{ID}(\varepsilon,\delta|\mathbf{W})&:=\sup\{R|\text{R is }(\varepsilon,\delta)\text{-achievable}\} \label{eq:id-capacity-1}\\
    &:=\lim\inf_{n\to\infty}\frac{1}{n}\log\log{N^*(\varepsilon,\delta|\mathbf{W})}. \label{eq:id-capacity-2}
\end{align}
\mmchange{The ID capacity is defined as the $(\varepsilon,\delta)$-ID capacity with $\varepsilon=0$ and $\delta=0$, i.e., $C_{ID}(\mathbf{W}):=C_{ID}(0,0|\mathbf{W})$.}
\end{definition}
An achievable bound of $C_{ID}(\varepsilon,\delta|\mathbf{W})$ is stated as follows.
\begin{theorem}[{\cite{hayashi}}]
    \label{thm:id-capacity-achievable-general}
    For $0\leq\varepsilon$, $\delta<1$, $\varepsilon+\delta<1$ and a sequence $\textbf{W} = (W^n)_{n\in\mathbb{N}}$ of general channels, we have
    \begin{equation}
        C_\mathrm{ID}(\varepsilon,\delta|\textbf{W}) \geq \sup_{\textbf{X}} \underline{I}^{\varepsilon}(\textbf{X};\textbf{Y}).
    \end{equation}
\end{theorem}
\mmchange{In Section \ref{id}, we will present two converse proofs. The first one, presented in Section \ref{sec:hayashi-converse}, only matches the achievability bound in Theorem \ref{thm:id-capacity-achievable-general} under the assumption of the \textit{strong converse property}. The second one, presented in Section \ref{sec:converse_technique_by_watanabe} matches the bound in Theorem \ref{thm:id-capacity-achievable-general} for $\delta = 0$ without further assumptions, establishing the identification capacity of a general (discrete) channel.}

\subsection{Problem Formulations of Transmission or Identification in the presence of feedback}
\label{problem_idf}
In this section, we formulate the problem of transmission or identification in the presence of noiseless feedback (IDF) and review recent results.

The problem of message transmission with noiseless feedback is illustrated in Fig \ref{fig:transmission_feedback_diagramm}.
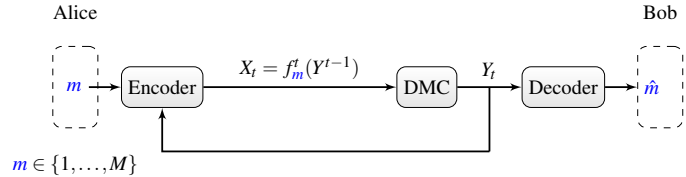
\begin{figure}[!ht]
    \centering
    \scalebox{0.92}{
\tikzstyle{farbverlauf} = [ top color=white, bottom color=white!80!gray]
\tikzstyle{block} = [draw,top color=white, bottom color=white!80!gray, rectangle, rounded corners,
minimum height=2em, minimum width=2.5em]

\tikzstyle{blocked} = [draw, fill=none, rectangle, rounded corners,
minimum height=4em, minimum width=.7
cm]
\tikzstyle{block1} = [draw, fill=none, rectangle, rounded corners,
minimum height=4em, minimum width=2
cm]
\tikzstyle{input} = [coordinate]
\tikzstyle{sum} = [draw, circle,inner sep=0pt, minimum size=2mm,  thick]

\scalebox{.93}{
\tikzstyle{arrow}=[draw,->] 
\begin{tikzpicture}[auto, node distance=2cm,>=latex']
\node[] (M) {${\color{blue}{m}}$};
\node[blocked, dashed, left=-0.6cm of M] (alice) {};
\node[block,right=.5cm of M] (enc) {Encoder};
\node[block, right=3cm of enc] (channel) {DMC};
\node[block, right=1cm of channel] (dec) {Decoder};
\node[blocked, dashed, right=0.5 cm of dec] (bob) {};
\node[align=left,right=.5cm of dec] (Mhat) {\color{blue}{$\hat{m}$}};
\node[input,right=.5cm of channel] (t1) {};
\node[input,below=1cm of t1] (t2) {};
\node[above= 0.3 cm of alice] (a) {Alice};
\node[above= 0.3 cm of bob] (b) {Bob};
\node[below=0.3 cm of alice] (message) {${\color{blue}{m}} \in \{1,\ldots,M\}$};
\draw[->,thick] (M) -- (enc);
\draw[->,thick] (enc) --node[above]{ $X_t={f}^t_{\color{blue}{m}}(Y^{t-1})$} (channel);
\draw[->,thick] (channel) --node[above]{$Y_t$} (dec);
\draw[->,thick] (dec) -- (Mhat);
\draw[-,thick] (t1) -- (t2);
\draw[->,thick] (t2) -| (enc);
\end{tikzpicture}}
}
    \caption{Transmission with noiseless feedback over a DMC}
    \label{fig:transmission_feedback_diagramm}
    \end{figure}
Given an input message $i$ from a message set $\mathcal{M}=\{1,...,M\}$ with cardinally $M$ and a DMC $W$, the encoder maps the feedback sequences $y^{t-1}=y_1,\cdots, y_{t-1}$ using the feedback encoding functions $\boldsymbol{f}_i$, and then sends $f_i^t(y^{t-1})$ over the DMC $W$. The decoder outputs an estimate $\hat{i}\in\mathcal{M}$ of the message sent. A transmission feedback code is defined as follows:
\begin{definition}
    An $(n,M,\varepsilon)$ transmission feedback code for the DMC $W$ is a system $\left\{\boldsymbol{f}_i,\mathcal{D}_i\right\}_{i=1}^{M}$ with
    \begin{align}
        \boldsymbol{f}_i\in\mathcal{F}_n,\quad \mathcal{D}_i\subset\mathcal{Y}^n, \quad \forall i\in\mathcal{M},
    \end{align}
    and an error probability that satisfies
    \begin{align}
    \label{eq:error}
        P_e(i)=Pr[i\ne\hat{i}]\le \varepsilon,
    \end{align}
    for all $i=1,\ldots,M$ and some $\varepsilon \in (0,1)$.
\end{definition}

Wolfowitz proved that feedback can not increase the transmission capacity of DMCs\cite{wolfowitz}.

Consider the IDF problem via single-user DMC $W$, as illustrated in Figure \ref{fig:idf_code}. The message $i$ from the message set $\mathcal{N}=\{1,...,N\}$ with cardinality $|\mathcal{N}|=N$ is encoded by a feedback encoding function $\boldsymbol{f}_i$ with respect to message $i$ and is transmitted over the channel $W$. A feedback encoding function is defined as follows.
\begin{definition}
A feedback encoding function $\boldsymbol{f}_i$ w.r.t. message $i\in\mathcal{N}$ is a vector-valued function
\begin{align}
    \boldsymbol{f}_i=[f_i^1,\cdots,f_i^n],
\end{align}
where $f_i^1\in\mathcal{X}$ and $f_i^t:\mathcal{Y}^{t-1}\mapsto\mathcal{X}$. We denote the set of all feedback encoding functions with length $n$ as $\mathcal{F}_n$.
\end{definition}
The output $Y_t$ of the channel is returned via a feedback loop connected to the encoder, which is represented by the feedback strategy $\boldsymbol{f}$. 
\begin{figure}[!ht]
    \centering  
    \scalebox{0.85}{
\tikzstyle{farbverlauf} = [ top color=white, bottom color=white!80!gray]
\tikzstyle{block} = [draw,top color=white, bottom color=white!80!gray, rectangle, rounded corners,
minimum height=2em, minimum width=2.5em]

\tikzstyle{blocked} = [draw, fill=none, rectangle, rounded corners,
minimum height=4em, minimum width=.7
cm]
\tikzstyle{block1} = [draw, fill=none, rectangle, rounded corners,
minimum height=4em, minimum width=2
cm]
\tikzstyle{input} = [coordinate]
\tikzstyle{sum} = [draw, circle,inner sep=0pt, minimum size=2mm,  thick]

\scalebox{.93}{
\tikzstyle{arrow}=[draw,->] 
\begin{tikzpicture}[auto, node distance=2cm,>=latex']
\node[] (M) {${\color{blue}{i}}$};
\node[blocked, dashed, left=-0.6cm of M] (alice) {};
\node[block,right=.5cm of M] (enc) {Encoder};
\node[block, right=3cm of enc] (channel) {DMC};
\node[block, right=1cm of channel] (dec) {Decoder};
\node[block1, dashed, right=0.5 cm of dec] (bob) {};
\node[align=left,right=.5cm of dec] (Mhat) {Is ${i^\prime}$ sent? \\ Yes or No?};
\node[input,right=.5cm of channel] (t1) {};
\node[input,below=1cm of t1] (t2) {};
\node[above= 0.3 cm of alice] (a) {Alice};
\node[above= 0.3 cm of bob] (b) {Bob};
\node[below=0.3 cm of alice] (message) {${\color{blue}{i}} \in \{1,\ldots,N\}$};
\draw[->,thick] (M) -- (enc);
\draw[->,thick] (enc) --node[above]{ $X_t={f}^t_{\color{blue}{i}}(Y^{t-1})$} (channel);
\draw[->,thick] (channel) --node[above]{$Y_t$} (dec);
\draw[->,thick] (dec) -- (Mhat);
\draw[-,thick] (t1) -- (t2);
\draw[->,thick] (t2) -| (enc);
\end{tikzpicture}}
}
    \caption{Identification with noiseless feedback over a DMC}
    \label{fig:idf_code}
\end{figure}
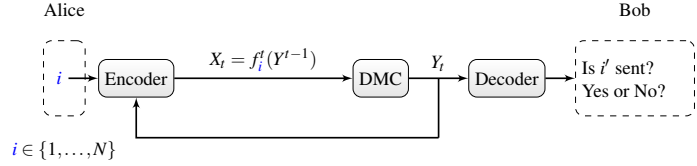

In the following, we revisit the definition of deterministic and randomized IDF codes for DMCs, as introduced in \cite{ahlswede_feedback}, respectively.
\begin{definition}
    An $(n,N,\lambda)$ deterministic IDF code for a DMC $W$ is a system $\left\{\boldsymbol{f}_i,\mathcal{D}_i)\right\}_{i=1}^N$, where
    \begin{align}
        \boldsymbol{f}_i\in\mathcal{F}_n,\quad \mathcal{D}_i\subset \mathcal{Y}^n,
    \end{align}
and for all $i,j\in\mathcal{N}$ with $i\ne j$ and some $\lambda \in (0,\frac{1}{2})$, where the type I and type II error satisfy
\begin{align}
    P_I(i)&=W^n(\mathcal{D}^c|\boldsymbol{f}_i)\le \epsilon\leq \lambda,\\
    P_{II}(i,j)&=W^n(\mathcal{D}_j|\boldsymbol{f}_i)\le \delta\leq \lambda.
\end{align}
\end{definition}
\begin{definition}
    An $(n,N,\lambda)$ randomized IDF code for a DMC $W$ is a system $\left\{Q_F(\cdot|i),\mathcal{D}_i\right\}$, where
    \begin{align}
        Q_F(\cdot|i)\in\mathcal{P}(\mathcal{F}_n),\quad \mathcal{D}_i\subset\mathcal{Y}^n,
    \end{align}
   and for all $i,j\in\mathcal{N}$ with $i\ne j$ and some $\lambda \in (0, \frac{1}{2})$, the type I and type II error satisfy
    \begin{align}
        P_I(i)&=\sum_{\boldsymbol{g}\in\mathcal{F}_n}Q_F(\boldsymbol{g}|i)W^n(\mathcal{D}_i^c|\boldsymbol{g})\le \lambda,\\
    P_{II}(i,j)&=\sum_{\boldsymbol{g}\in\mathcal{F}_n}Q_F(\boldsymbol{g}|i)W^n(\mathcal{D}_j|\boldsymbol{g})\le \lambda.
    \end{align}
\end{definition}
\mmchange{Let $N_f(n, \lambda)$ and $N_F(n, \lambda)$ be the maximal integers $N$ such that an $(n, N, \lambda)$ deterministic or randomized IDF code exists, respectively.}
We define the rate of both deterministic and randomized IDF codes as \mmchange{double logarithmic}, i.e.,
\begin{align}
    R=\frac{1}{n}\log\log{N}.
    \label{eq.R_IDF}
\end{align}
An IDF rate $R$ is achievable if, for some $\lambda\in(0,\frac{1}{2})$, there's an $(n,2^{2^{nR}},\lambda)$ IDF code. The deterministic and randomized IDF capacities $C_{dIDF}(W)$ and $C_{rIDF}(W)$ are the supremum of all achievable rates.
Although feedback does not increase the transmission capacity of DMCs, Ahlswede and Dueck proved in \cite{ahlswede_feedback}, that feedback increases the identification capacity via noisy channels. They examined the IDF capacity formulas with deterministic and randomized encoders, which are given in Theorem \ref{thm:deterministic_encoder} and \ref{thm:randomized_encoder}, respectively. 
\begin{theorem}[\cite{ahlswede_feedback}]
\label{thm:deterministic_encoder}
    Let $C(W)$ be the Shannon capacity of the DMC $W$. Then the deterministic ID feedback capacity $C_{dIDF} (W)$ of the DMC $W$ is given by
    \begin{equation}
        C_{dIDF} = 
        \begin{cases}
            0, &\text{ if } C(W)=0 \text{ or } W \text{ is noiseless} \\
            \max\limits_{x\in\mathcal{X}} H(W(\cdot|x)), &\text{ otherwise }.
        \end{cases}
    \end{equation}
   
\end{theorem}

\begin{theorem}[\cite{ahlswede_feedback}]
\label{thm:randomized_encoder}
 Let $C(W)$ be the Shannon capacity of the DMC $W$. Then the  randomized IDF capacity $C_{rIDF}(W)$ of the DMC $W$ is given by   
    \begin{equation}
        C_{rIDF} = 
        \begin{cases}
            C(W), &\text{ if } C=0 \text{ or } W \text{ is noiseless} \\
            \max\limits_{P \in \mathcal{P}(\mathcal{X})} H(P \cdot W), &\text{otherwise},
        \end{cases}
    \end{equation}
    where $P\cdot W(y)=\sum_{x\in\mathcal{X}}P(x)W(y|x)$ for all $y\in\mathcal{Y}$.
\end{theorem}

\section{Strong Converse for Identification via Channels}
\label{id}
In this section, we will introduce the strong converse for the identification capacity based on channel resolvability. We will skip the first formulation by Han and Verdú \cite{han_verdu}, which is also presented in detail in \cite{han}. We will focus on two newer and related formulations of the converse proof. \mmchange{First, we will briefly introduce the concepts of channel resolvability and partial channel resolvability in Section \ref{sec:channel_resolvability}. In Section \ref{sec:hayashi-converse}, we present the converse based on channel resolvability} due to Hayashi \cite{hayashi} and Oohama \cite{oohama}. They proved the strong converse, but it only holds for channels where the condition $\overline{I}(X;Y) = \underline{I}(X;Y)$ \mmchange{is fulfilled}. \mmchange{The proof we present in Section \ref{sec:converse_technique_by_watanabe}} is due to Watanabe \cite{watanabe}, who was able to strengthen the proof and apply it to general channels without the strong converse condition. 

\subsection{Channel Resolvability}
\label{sec:channel_resolvability}
The idea of channel resolvability first emerged for the converse of identification via channels \cite{han_paper} and was later studied on its own \cite{han_verdu}. A fundamental concept for channel resolvability are $M$-type distributions $\tilde{P}$.

\begin{definition}
    For a positive integer $M$, a probability distribution $\tilde{P}$ is called \emph{$M$-type} if
    \begin{align}
        \tilde{P}(x) \in \left\{0, \frac{1}{M},\frac{2}{M},...,1\right\} \qquad \text{for all } x \in \mathcal{X}.
    \end{align}
\end{definition}
The number of $M$-type distributions is upper bounded by $|\mathcal{X}|^M$ \cite{han_verdu}.

In channel resolvability, the output distribution $PW$ of a channel $W$ induced by an input distribution $P$ should be approximated by the output distribution $\tilde{P}W$ induced by the $M$-type $\tilde{P}$ such that the variational distance between the two is bounded within an allowed approximation error $\zeta$ satisfying $0 \leq \zeta < 1$:
\begin{align}
    d(\tilde{P}W, PW) \leq \zeta. \label{eq:resolvability-condition}
\end{align}
The goal is to choose $M$ as small as possible such that \eqref{eq:resolvability-condition} is still fulfilled.

The requirement \eqref{eq:resolvability-condition} can be relaxed such that only a part of the output distribution must be accurately appoximated, which leads to the concept of \textit{partial channel resolvability} \cite{steinberg}. For any subset $\mathcal{S} \subset \mathcal{X} \times \mathcal{Y}$ denote with
\begin{align}
    PW^{\mathcal{S}}(y) := \sum_{x \in \mathcal{X}} P(x)W(y|x) \mathds{1}[(x,y) \in \mathcal{S}] \label{eq:partial_response}
\end{align}
the \textit{partial response $W$ due to $P$ on $\mathcal{S}$} \cite{steinberg}. An important property which we will use several times is that, for arbitrary $P$, $W$, and $\mathcal{S}$, it holds that
\begin{align}
    d(PW, PW^{\mathcal{S}}) &= \frac{1}{2}\left|\sum_{y \in \mathcal{Y}}\sum_{x \in \mathcal{X}} P(x)W(y|x)- P(x)W(y|x) \mathds{1}[(x,y) \in \mathcal{S}]\right| \notag\\
    &= \frac{1}{2}\sum_{y \in \mathcal{Y}}\sum_{x \in \mathcal{X}} P(x)W(y|x) \mathds{1}[(x,y) \in \mathcal{S}^c] \notag \\
    &= \frac{1}{2}\sum_{y \in \mathcal{Y}}PW^{\mathcal{S}^c}(y) \label{eq:var-distance-partial-response-1}\\
    &= \frac{1}{2}P \times W(\mathcal{S}^c), \label{eq:var-distance-partial-response-2}
\end{align}

The goal of partial channel resolvability is to approximate the output distribution $PW^{\mathcal{S}}$ for a given set $\mathcal{S}$ by using an $M$-type input:
\begin{align}
    d(\tilde{P}W^{\mathcal{S}}, PW^{\mathcal{S}}) \leq \zeta. \label{eq:partial-resolvability-condition}
\end{align}
Since we will use the resolvability only as an auxiliary tool to prove the converse of the identification capacity, we will not go into details with definitions of \textit{achievable resolution rates} or the \textit{resolvability}, which is the smallest achievable resolvability rate. For further details, see \cite{han_verdu}, \cite{han}.

\subsection{Converse based on Channel Resolvability}\label{sec:hayashi-converse}
The resolvability-based converse proof consists of two main steps. The first step is to find a bound for the channel resolvability problem \eqref{eq:resolvability-condition}, i.e., to find a connection between the $M$-type size $M$ and the variational distance between the output distributions. This is done in the so-called \textit{soft covering lemma} (Lemma \ref{lem:soft_covering_lemma_hayashi}). In a second step, we will connect the channel resolvability problem to the problem of identification via channels and use the soft covering lemma together with a type-counting argument to prove a converse bound for the identification via channels problem.

\begin{lemma}[Soft Covering Lemma \cite{hayashi},\cite{oohama}]
\label{lem:soft_covering_lemma_hayashi}
    Let $\gamma$ be any real constant and define the set
    \begin{equation}
        \mathcal{T} = \hayashiset = \left\{ (x,y) \in \mathcal{X} \times\mathcal{Y} | \log\frac{W(y|x)}{PW(y)} \leq \gamma \right\}.\label{eq:definition_T}
    \end{equation}
    For a given input $P$, there exists an $M$-type input $\Tilde{P}$ such that
    \begin{equation}
    \label{eq:soft_covering_hayashi}
        d(\Tilde{P}W, PW) \leq P\times W(\mathcal{T}^c) + \frac{1}{2} \sqrt{\frac{e^\gamma}{M}}.
    \end{equation}
\end{lemma}

In the following, we will typically use the shorthand notation $\mathcal{T}$ to refer to the set $\hayashiset$, unless the dependency on $P$ and $\gamma$ mus be explicitly stated. The following formulation of the proof follows Oohama \cite{oohama}.
\begin{proof}
    For the proof, we use a random coding argument. Consider a codebook $\mathcal{C} = \{X_1, ..., X_M\}$, where each codeword is randomly generated with distribution $P$. We now define the $M$-type $\tilde{P}$ as
    \begin{align}
        \tilde{P}(x) = \frac{1}{M}\sum_{i = 1}^M \mathds{1}[X_i = x].
    \end{align}
    Therefore, we have
    \begin{align}
        \tilde{P}W(y) &= \sum_{x\in\mathcal{X}^n}\tilde{P}(x)W(y|x) = \sum_{x\in\mathcal{X}^n}\frac{1}{M}\sum_{i = 1}^M \mathds{1}[X_i = x]W(y|x)
    \end{align}
    The proof is structured in three steps.
    \begin{enumerate}
        \item \textit{Decomposing the variational distance} $d(\Tilde{P}W, PW)$: We use the partial response of the channel on the set $\mathcal{T}$ and the triangular inequality to obtain
        \begin{align}
            d(\Tilde{P}W, PW) &{\leq} d(\tilde{P}W, \tilde{P}W^{\mathcal{T}}) + d(\tilde{P}W^{\mathcal{T}}, PW^{\mathcal{T}}) +  d(PW^{\mathcal{T}}, PW) \notag\\
            &\stackrel{\eqref{eq:var-distance-partial-response-1}}{=} \frac{1}{2}\sum_{y \in \mathcal{Y}}\left[PW^{\mathcal{T}^c}(y) + \tilde{P}W^{\mathcal{T}^c}(y)\right] + d(\tilde{P}W^{\mathcal{T}}, PW^{\mathcal{T}}).\label{eq:proof-scl-hayashi-decompose-vd}
        \end{align}
    
        \item \textit{Random coding argument:} We now take the expectation of \eqref{eq:proof-scl-hayashi-decompose-vd} over the codebook $\mathcal{C}$. If the expectation can be bounded by some value, we know that there must exist at least one codebook where the same bound on the variational distance must hold. We note that 
        \begin{align}
        \mathbb{E}[\Tilde{P}W^\mathcal{T}(y)] &= \sum\limits_{X_1 \in \mathcal{X}^n} P(X_1) \cdots\sum\limits_{X_M \in \mathcal{X}^n} P(X_M) \sum_{x\in\mathcal{X}^n}\frac{1}{M}\sum_{i = 1}^M \mathds{1}[X_i = x]W(y|x) \notag\\
        &= \sum\limits_{i=1}^{M} \frac{1}{M} \sum\limits_{x\in \mathcal{X}^n} P(x)W^\mathcal{T}(y|x) \notag\\
        &= PW^\mathcal{T}(y).\label{eq:proof_hayashi_2}
    \end{align}
    Taking the expectation on \eqref{eq:proof-scl-hayashi-decompose-vd} and the linearity of the expectation yields
    \begin{align}
        \mathbb{E}[d(\Tilde{P}W,PW)] \leq& \frac{1}{2} \sum\limits_{y\in \mathcal{Y}^n} \underbrace{\mathbb{E}[\Tilde{P}W^{\mathcal{T}^c}(y)]}_{= PW^{\mathcal{T}^c}(y)} + \underbrace{\mathbb{E}[PW^{\mathcal{T}^c}(y)]}_{= PW^{\mathcal{T}^c}(y)} +\mathbb{E}\left[d(\tilde{P}W^{\mathcal{T}}, PW^{\mathcal{T}})\right]. \notag\\
        =& \sum\limits_{y\in \mathcal{Y}^n} PW^{\mathcal{T}^c}(y) + \mathbb{E}\left[d(\tilde{P}W^{\mathcal{T}}, PW^{\mathcal{T}})\right]\notag\\
        =& P\times W(\mathcal{T}^c)+ \mathbb{E}\left[d(\tilde{P}W^{\mathcal{T}}, PW^{\mathcal{T}})\right]\label{eq:proof-scl-hayashi-random-coding}
    \end{align}
    \item \textit{Upper bounding the remaining variational distance:} We now aim to find an upper bound for the second term in \eqref{eq:proof-scl-hayashi-random-coding}. 
    \begin{align}
        \mathbb{E}\left[ d(\tilde{P}W^{\mathcal{T}}, PW^{\mathcal{T}})\right] &=  \sqrt{\mathbb{E}\left[\frac{1}{2} \sum\limits_{y\in \text{supp}(PW^\mathcal{T})}|\Tilde{P}W^\mathcal{T}(y)-PW^\mathcal{T}(y)|\right]^2} \notag\\
        &\leq \sqrt{\mathbb{E}\left[\left(\frac{1}{2} \sum\limits_{y\in \text{supp}(PW^\mathcal{T})}|\Tilde{P}W^\mathcal{T}(y)-PW^\mathcal{T}(y)|\right)^2\right]} \label{eq:proof-scl-hayashi-vd-bound-variance-jensen} \\
        &= \sqrt{\frac{1}{4}\mathbb{E}\left[\left(\sum\limits_{y\in \text{supp}(PW^\mathcal{T})}\mkern-18mu \sqrt{PW^{\mathcal{T}}(y)}\frac{|\Tilde{P}W^\mathcal{T}(y)-PW^\mathcal{T}(y)|}{\sqrt{PW^{\mathcal{T}}(y)}}\right)^2\right]} \notag\\
        &\leq \sqrt{\frac{1}{4}\mathbb{E}\left[\underbrace{\sum\limits_{y\in \text{supp}(PW^\mathcal{T})}\mkern-18mu PW^{\mathcal{T}}(y)}_{=1} \mkern-18mu \sum\limits_{y\in \text{supp}(PW^\mathcal{T})} \mkern-18mu \frac{|\Tilde{P}W^\mathcal{T}(y)-PW^\mathcal{T}(y)|^2}{PW^{\mathcal{T}}(y)}\right]} \label{eq:proof-scl-hayashi-vd-bound-variance-csi}\\
        &=  \sqrt{\frac{1}{4}\sum\limits_{y\in \text{supp}(PW^\mathcal{T})}\frac{\text{Var}\left[\Tilde{P}W^\mathcal{T}(y)\right]}{PW^{\mathcal{T}}(y)}} \label{eq:proof-scl-hayashi-vd-bound-variance}
    \end{align}
    The inequality \eqref{eq:proof-scl-hayashi-vd-bound-variance-jensen} follows from Jensen's inequality and \eqref{eq:proof-scl-hayashi-vd-bound-variance-csi} follows from the Cauchy-Schwarz inequality.
    \item \textit{Upper bounding the variance:} By definition of the variance, we have
    \begin{equation}
        \text{Var}[\Tilde{P}W^\mathcal{T}(y)] = \mathbb{E}[\Tilde{P}W^\mathcal{T}(y)^2]-\Big(\mathbb{E}[\Tilde{P}W^\mathcal{T}(y)]\Big)^2. \label{eq:proof-scl-hayashi-variance-def}
    \end{equation}
    Therefore we consider
    \begin{align}
        \Tilde{P}W^\mathcal{T}(y)^2 &= \left(\sum\limits_{i=1}^{M}\frac{1}{M} W^\mathcal{T}(y|X_i)\right)^2 \notag \\
        &= \frac{1}{M^2} \sum\limits_{i=1}^{M} W^\mathcal{T}(y|X_i)^2 + \frac{1}{M^2} \sum\limits_{\substack{i,j=1\\i\neq j}}^{M} W^\mathcal{T}(y|X_i) W^\mathcal{T}(y|X_j). 
    \end{align}
    Now we take the expectation over the codebook $\{X_1, ..., X_M\}$ on both sides. Note that 
    \begin{align}
        \mathbb{E}\left[W^\mathcal{T}(y|X_i)\right] = PW^\mathcal{T}(y).
    \end{align}
    and 
    \begin{align}
        \mathbb{E}[W^\mathcal{T}(y|X_i) W^\mathcal{T}(y|X_j)] &= \mathbb{E}[W^\mathcal{T}(y|X_i) ]\mathbb{E}[W^\mathcal{T}(y|X_j)]\\
        &= PW^\mathcal{T}(y)^2.
    \end{align}
    Thus, we get
    \begin{align}
        \mathbb{E}[\big(\Tilde{P}W^\mathcal{T}(y)\big)^2] &=  \frac{1}{M^2} \sum\limits_{i=1}^{M} \mathbb{E}[W^\mathcal{T}(y|X_i)^2] + \frac{1}{M^2} \sum\limits_{\substack{i,j=1\\i\neq j}}^{M} \mathbb{E}[W^\mathcal{T}(y|X_i) W^\mathcal{T}(y|X_j)] \notag \\
        & = \frac{1}{M^{2}} \cdot M \sum\limits_{x} P(x) W^\mathcal{T}(y|x)^2 + \frac{1}{M^2} \cdot M\cdot(M-1) PW^\mathcal{T}(y)^2 \notag\\
        &\leq \frac{1}{M}  \sum\limits_{x} P(x) W^\mathcal{T}(y|x)^2 + PW^\mathcal{T}(y)^2 \label{eq:proof-scl-hayashi-second-moment}
    \end{align}
    Inserting \eqref{eq:proof-scl-hayashi-second-moment} and \eqref{eq:proof_hayashi_2} into \eqref{eq:proof-scl-hayashi-variance-def}, we obtain
    \begin{align}
        \text{Var}[\Tilde{P}W^\mathcal{T}(y)] &\leq \frac{1}{M}  \sum\limits_{x} P(x) W^\mathcal{T}(y|x)^2 + PW^\mathcal{T}(y)^2 - PW^\mathcal{T}(y)^2 \notag\\
        &= \frac{1}{M}  \sum\limits_{x} P(x) W^\mathcal{T}(y|x)^2 \notag\\
        &\stackrel{\eqref{eq:partial_response}, \eqref{eq:definition_T}}{=} \frac{1}{M}  \sum\limits_{x} P(x) W(y|x)^2 \cdot \mathds{1}\left[\log\frac{W(y|x)}{PW(y)} \leq \gamma \right] \notag\\
        &\leq \frac{1}{M}  \sum\limits_{x} P(x) W(y|x) PW(y) \cdot e^{\gamma} \cdot \mathds{1}\left[\log\frac{W(y|x)}{PW(y)} \leq \gamma \right] \label{eq:proof-scl-hayashi-variance-bound}
    \end{align}
    \item \textit{Completion:} We can now insert the upper bound on the variance \eqref{eq:proof-scl-hayashi-variance-bound} into \eqref{eq:proof-scl-hayashi-vd-bound-variance}:
    \begin{align}
        \mathbb{E}\left[d(\tilde{P}W^{\mathcal{T}}, PW^{\mathcal{T}})\right] 
        &\leq \sqrt{\frac{1}{4}\sum\limits_{y\in \mathcal{Y}^n}\frac{\frac{1}{M}  \sum\limits_{x\in\mathcal{X}^n} P(x) W(y|x) PW^\mathcal{T}(y) \cdot e^{\gamma}}{PW^{\mathcal{T}}(y)}} \notag\\
        &= \sqrt{\frac{1}{4M}\cdot e^{\gamma}\sum\limits_{y\in \mathcal{Y}^n}\sum\limits_{x\in\mathcal{X}^n}   P(x) W(y|x)} \notag\\
        &= \frac{1}{2}\sqrt{\frac{e^{\gamma}}{M}} 
    \end{align}
    Together with \eqref{eq:proof-scl-hayashi-random-coding}, we obtain
    \begin{align}
       \mathbb{E}[d(\Tilde{P}W,PW)] &\leq P\times W(\mathcal{T}^c) + \frac{1}{2} \sqrt{\frac{e^\gamma}{M}}. 
    \end{align}
    \end{enumerate}
\end{proof} 

Lemma \ref{lem:soft_covering_lemma_hayashi} can be used to prove the achievability part of channel resolvability: it can be shown that, as long as $\frac{\log M}{n}$ is larger than the capacity, the expression on the right hand side vanishes asymptotically for a proper choice of $\gamma$. We are, however, interested in the application of Lemma \ref{lem:soft_covering_lemma_hayashi} to identification via channels. The next lemma will create that connection. The fundamental idea is that identification codes whose codeword distributions are $M$-types are limited in size by the total number of $M$-types. We then use Lemma \ref{lem:soft_covering_lemma_hayashi} to approximate arbitrary codeword distributions by $M$-types, resulting in a bound for arbitrary ID-codes.

\begin{lemma}[{\cite[Lemma 3]{hayashi}}]
\label{lem:id-error-bound-hayashi}
    Let $\mathcal{T}$ be defined as in \eqref{eq:definition_T} for any $\gamma\in\mathbb{R}$. For an integer $M$, any ($N$, $\epsilon$, $\delta$)-ID-code with $N > |\mathcal{X}|^M$ must satisfy
    \begin{equation}
        \varepsilon+\delta\geq \inf\limits_{P} \left\{1-2 P\times W(\mathcal{T}^c)\right\} -\sqrt{\frac{e^\gamma}{M}}.
    \end{equation}
\end{lemma}
\begin{proof}
Conisder a given ($N$, $\epsilon$, $\delta$)-ID code $\{(P_i, \mathcal{D}_i)\}_{i=1}^N$.
\begin{enumerate}
    \item \textit{Relation between the variational distance and the error probabilities}: 
    Note that $P_iW(\mathcal{D}_i) = 1- P_iW(\mathcal{D}_i^c)$, where $P_iW(\mathcal{D}_i^c)$ denotes the type-I error probability and $P_jW(\mathcal{D}_i)$ denotes the type-II error probability for $i\neq j$. Therefore, $P_iW(\mathcal{D}_i) \geq 1-\varepsilon$ and $P_jW(\mathcal{D}_i)\leq \delta$ must hold. We can now connect the error probabilities of the ID code to the variational distance of the output distributions of different inputs:
\begin{align}
    d(P_iW,P_jW) &\stackrel{\eqref{eq:variational-distance}}{=} \sup_{\mathcal{A} \subseteq \mathcal{Y}} |P_iW(\mathcal{A}) - P_jW(\mathcal{A})|\notag\\
    &\geq P_iW(\mathcal{D}_i) - P_jW(\mathcal{D}_i) \notag\\
    &\geq 1-\varepsilon - \delta \label{eq:hayashi-vd-error-bound}
\end{align}
    for every $i \neq j$. 
    \item \textit{Replacing the input by $M$-types.} We now replace the input $P_i$ for each message $i$ by an $M$-type distribution $\tilde{P}_i$. Lemma \ref{lem:soft_covering_lemma_hayashi} guarantees that we can find an $M$-type distribution such that the variational distance of the output distributions is bounded by 
    \begin{equation}
        d(\Tilde{P}_iW, P_iW) \leq P\times W(\mathcal{T}_{P_{i}}(\gamma)^c) + \frac{1}{2} \sqrt{\frac{e^\gamma}{M}}. \label{eq:hayashi-vd-codeword-bound}
    \end{equation}
    \item \textit{Non-distinctness of inputs.} \label{step:proof-hayashi-non-distinct} Notice that the number of distinct $M$-types is upper bounded by $|\mathcal{X}|^M$. However, by assumption, we have $N \geq |\mathcal{X}|^M$, and therefore, there must exist a pair $i'$ and $j'$ such that $\tilde{P}_{i'} = \tilde{P}_{j'}$. For this pair, it immediately follows that also
    \begin{align}
        d(\Tilde{P}_{i'}W,\Tilde{P}_{j'}W) = 0\label{eq:hayashi-vd-output-same-input}
    \end{align}
    must hold.
    \item \textit{Estimation of the variational distance with the triangular equality}:
    We will now use \eqref{eq:hayashi-vd-codeword-bound} to find an upper bound on $d(P_{i'}W,P_{j'}W)$ for the pair $(i', j')$ from step \ref{step:proof-hayashi-non-distinct}. To this end, we apply the triangular inequality:
\begin{align}
    d(P_{i'}W,P_{j'}W) &\leq d(P_{i'}W,\Tilde{P}_{i'}W) + d(\Tilde{P}_{i'}W,\Tilde{P}_{j'}W)+ d(\Tilde{P}_{j'}W,P_{j'}W) \notag \\
    &\stackrel{\eqref{eq:hayashi-vd-codeword-bound}, \eqref{eq:hayashi-vd-output-same-input}}{\leq}  P\times W(\mathcal{T}_{P_{i'}}(\gamma)^c) +  P\times W(\mathcal{T}_{P_{j'}}(\gamma)^c) + \sqrt{\frac{e^\gamma}{M}}\notag\\
    &\leq \sup_P \left\{2P\times W(\mathcal{T}_{P}(\gamma)^c)\right\} + \sqrt{\frac{e^\gamma}{M}}.  \label{eq:hayashi-vd-different-output-prime}
\end{align}
    \item \textit{Completion.} Since \eqref{eq:hayashi-vd-error-bound} holds for every pair $(i,j)$, it must also hold for the specific pair $(i', j')$. Combining \eqref{eq:hayashi-vd-error-bound} and \eqref{eq:hayashi-vd-different-output-prime} results in
\begin{align}
    1-\varepsilon-\delta &\leq d(P_{i'}W,P_{j'}W) \\
    &\leq \sup_P \left\{2P\times W(\mathcal{T}_{P}(\gamma)^c)\right\} + \sqrt{\frac{e^\gamma}{M}}.
\end{align}
By rearranging the terms, we obtain
\begin{align}
    \varepsilon+\delta \geq \inf\limits_{P} \left\{1-2 P\times W(\mathcal{T}_{P}(\gamma)^c)\right\} -\sqrt{\frac{e^\gamma}{M}}.
\end{align}
\end{enumerate}
\end{proof}
\begin{remark}
    The structure of the converse proof is very similar to the original converse proof by Han and Verdú, which can be seen from the formulation of that proof in \cite[Lemma 6.4.1]{han}. The main difference is the usage of Lemma \ref{lem:soft_covering_lemma_hayashi}, which gives a sharper bound on the variational distance than \cite[Lemma 6.3.1]{han}. 
\end{remark}
In the first proof step of Lemma \ref{lem:id-error-bound-hayashi}, \eqref{eq:hayashi-vd-error-bound} imposes a lower bound on the minimum distance of the output distributions of an ID code. The crucial step in the proof is the observation that the total number of $M$-types is bounded, and in consequence, not all $M$-type inputs can be distinct. The minimum distance of the output distributions can therefore only be fulfilled through the approximation error, which follows from Lemma \ref{lem:soft_covering_lemma_hayashi}. If this approximation error is small, it immediately follows that the sum of the error probabilities of the ID code $\epsilon+\delta$ must be large.

As a consequence of Lemma \ref{lem:id-error-bound-hayashi}, it is possible to derive the following upper bound for the ID capacity. \mmchange{We refer the reader to \cite{oohama} for a proof based on the tools presented in this Section.}

\begin{theorem}[{\cite{han_verdu} \cite[Theorem 4]{oohama}}] 
    \label{thm:id_capacity_converse_sup}
    For a sequence $\mathrm{\mathbf{W}} = (W^n)_{n\in\mathbb{N}}$ of general channels, we have
    \begin{equation}
        C_\mathrm{ID}(\mathrm{\mathbf{W}}) \leq \sup_{\mathrm{\mathbf{X}}} \overline{I}(\mathrm{\mathbf{X}};\mathrm{\mathbf{Y}}).
    \end{equation}
\end{theorem}
\mmchange{When comparing Theorem \ref{thm:id_capacity_converse_sup} to the achievability bound in Theorem \ref{thm:id-capacity-achievable-general}, we observe that the two bounds only match under the condition}
\begin{align}
    \sup_{\mathrm{\mathbf{X}}} \underline{I}(\mathrm{\mathbf{X}};\mathrm{\mathbf{Y}}) = \sup_{\mathrm{\mathbf{X}}} \overline{I}(\mathrm{\mathbf{X}};\mathrm{\mathbf{Y}}).\label{eq:strong-converse-assumption}
\end{align}
\mmchange{This is equivalent to the \textit{strong converse property} (see \cite{han}, Definition 3.5.1 and Theorem 3.5.1).} In the next section, we will present a similar approach to the converse of the ID coding theorem, that will enable us to derive a converse theorem without the assumption \eqref{eq:strong-converse-assumption}.

\subsection{Converse based on Partial Channel Resolvability}
\label{sec:converse_technique_by_watanabe}

In this proof we will use the idea of \textit{partial channel resolvability}. As in the channel ressovability problem, we replace the input distribution $P$ by an $M$-type input distribution $\Tilde{P}$. In partial channel resolvability, however, the goal is that only the \textit{partial responses} \eqref{eq:partial_response} for some given set $\mathcal{S}$ are close to each other:
\begin{align}
    d(\tilde{P}W^{\mathcal{S}}, PW^{\mathcal{S}}) \leq \zeta.
\end{align}
The set $\mathcal{S}$ can be chosen, e.g. to be the set $\hayashiset$ in \eqref{eq:definition_T}. However, to obtain a stronger result, \cite{watanabe} follows \cite{zhang_covert_2021-1}, where the set in question is defined as 
\begin{equation}
    \mathcal{S} = \watanabeset := \left\{(x,y) \in \mathcal{X}\times \mathcal{Y}: \, \log\frac{W(y|x)}{Q(y)}\leq\gamma\right\}, \label{eq:s_truncated_channel}
\end{equation}
where $Q$ is an arbitrary output distribution. In the following, we will just use the notation $\mathcal{S}$ to refer to the set \eqref{eq:s_truncated_channel} unless the dependency on $Q$ and $\gamma$ must be made explicit. The difference of \eqref{eq:s_truncated_channel} to the set $\hayashiset$ in \eqref{eq:definition_T} is that an \textit{auxiliary output distribution} $Q$ is used in the denominator instead of using the channel output distribution $PW(y)$.  The following modification of the soft covering lemma for partial channel resolvability with the set \eqref{eq:s_truncated_channel} was proved as an intermediate result in \cite{zhang_covert_2021-1} and explicitly stated in \cite{watanabe}.

\begin{lemma}[Soft Covering Lemma with auxiliary output distribution {\cite[Lemma 10]{zhang_covert_2021-1}, \cite[Lemma 3]{watanabe}}]
    \label{lem:soft-covering-lemma-watanabe}
    Let the set $\mathcal{S}$ be given by \eqref{eq:s_truncated_channel} with an arbitrary auxiliary output distribution $Q$ and $\gamma\in\mathbb{R}$. For a given $P$, there exists an $M$-type $\Tilde{P}$ such that 
    \begin{equation}
        d(\Tilde{P}W^\mathcal{S},PW^\mathcal{S}) \leq \frac{1}{2} \sqrt{\frac{e^\gamma}{M}}.
    \end{equation}
\end{lemma}
The proof of Lemma \ref{lem:soft-covering-lemma-watanabe} is very similar to the proof of Lemma \ref{lem:soft_covering_lemma_hayashi}. The two major differences are
\begin{enumerate}
    \item Lemma \ref{lem:soft_covering_lemma_hayashi} is for channel resolvability, whereas Lemma \ref{lem:soft-covering-lemma-watanabe} is for partial channel resolvability. The version of Lemma \ref{lem:soft-covering-lemma-watanabe} for channel resolvability can be found in \cite[Lemma 10]{zhang_covert_2021-1}.
    \item Lemma \ref{lem:soft_covering_lemma_hayashi} uses the set $\mathcal{T}$ with the true output distribution in the denominator, while Lemma \ref{lem:soft-covering-lemma-watanabe} uses the set $\mathcal{S}$ with the auxiliary output distribution $Q$.
\end{enumerate}
The fact that we are now only interested in partial channel resolvability means that we can skip the first two steps in the proof of Lemma \ref{lem:soft_covering_lemma_hayashi}, which result in the term $P\times W(\mathcal{T}^c)$. In the following proof, we briefly show the main steps and highlight the similarities and differences to the proof of Lemma \ref{lem:soft_covering_lemma_hayashi}.
\begin{proof}
    \begin{enumerate}
    \item \textit{Random Coding:} As in Lemma \ref{lem:soft_covering_lemma_hayashi}, we use random coding to establish that, if $\mathbb{E}_{\mathcal{C}}\left[d(\Tilde{P}W^\mathcal{S},PW^\mathcal{S}) \right]$ can be bounded, then there must exist at least one codebook for which the same bound holds.
        \item \textit{Upper bounding the variational distance:} This step corresponds to step 3 in the proof of Lemma \ref{lem:soft_covering_lemma_hayashi}. We use Jensen's inequality and the Cauchy-Schwarz inequality in the same manner to obtain
    \begin{align}
        \mathbb{E}\left[d(\Tilde{P}W^\mathcal{S},PW^\mathcal{S}) \right]=\mathbb{E}&\left[ \frac{1}{2} \sum\limits_{y} |\Tilde{P}W^\mathcal{S}(y)-PW^\mathcal{S}(y)|\right] \notag \\
        &\leq  \sqrt{\frac{1}{4}\sum\limits_{y\in \mathcal{Y}^n}\frac{\text{Var}\left[\Tilde{P}W^\mathcal{T}(y)\right]}{Q(y)}}. \label{eq:proof-scl-watanabe-vd-bound-variance}
    \end{align}
    \item \textit{Upper bounding the variance:} We can re-use the result from \eqref{eq:proof-scl-hayashi-variance-bound} and just have to modify the last step, where we have to replace $PW(y)$ by $Q(y)$ due to the different definitions of $\mathcal{T}$ and $\mathcal{S}$:
    \begin{align}
        \text{Var}[\Tilde{P}W^\mathcal{S}(y)] &\leq \frac{1}{M}  \sum\limits_{x} P(x) W(y|x) Q(y) \cdot e^{\gamma} \cdot \mathds{1}\left[\log\frac{W(y|x)}{Q(y)} \leq \gamma \right]. \label{eq:proof-scl-watanabe-variance-bound}
    \end{align}
    \item \textit{Completion:} We can now insert the upper bound on the variance \eqref{eq:proof-scl-watanabe-variance-bound} into \eqref{eq:proof-scl-watanabe-vd-bound-variance}:
    \begin{align}
        \mathbb{E}\left[d(\Tilde{P}W^\mathcal{S},PW^\mathcal{S}) \right] 
        &\leq \sqrt{\frac{1}{4}\sum\limits_{y\in \mathcal{Y}^n}\frac{\frac{1}{M}  \sum\limits_{x\in\mathcal{X}^n} P(x) W(y|x) Q(y) \cdot e^{\gamma}\cdot \mathds{1}\left[\log\frac{W(y|x)}{Q(y)} \leq \gamma \right]}{Q(y)}} \notag\\
        &= \sqrt{\frac{1}{4M}\cdot e^{\gamma}\sum\limits_{y\in \mathcal{Y}^n}\sum\limits_{x\in\mathcal{X}^n}   P(x) W(y|x)\cdot \mathds{1}\left[\log\frac{W(y|x)}{Q(y)} \leq \gamma \right]} \notag\\
        &\leq \sqrt{\frac{1}{4M}\cdot e^{\gamma}\sum\limits_{y\in \mathcal{Y}^n}\sum\limits_{x\in\mathcal{X}^n}   P(x) W(y|x)} \notag\\
        &= \frac{1}{2}\sqrt{\frac{e^{\gamma}}{M}}.
    \end{align}
    \end{enumerate}
\end{proof}

The introduction of the auxiliary output distribution $Q$ enables a strengthening of Lemma \ref{lem:id-error-bound-hayashi}, which is stated in the following.

\begin{lemma}[{\cite[Theorem 1]{watanabe}}]
    \label{lem:id-error-bound-watanabe}
    Let $Q$ be an arbitrarily given output distribution, $\gamma \in \mathbb{R}$ and $\mathcal{S}$ be defined as in (\ref{eq:s_truncated_channel}). Then for an arbitrary integer M, any $(N, \varepsilon, \delta)$-ID-code with $N>|\mathcal{X}|^M$ must satisfy
    \begin{equation}
        \varepsilon + \delta \geq \inf_{P} P\times W(\mathcal{S}) - \sqrt{\frac{e^\gamma}{M}}. 
    \end{equation}
\end{lemma}

\begin{proof}
Conisder a given ($N$, $\epsilon$, $\delta$)-ID code $\{(P_i, \mathcal{D}_i)\}_{i=1}^N$.
\begin{enumerate}
    \item \textit{Relation between the variational distance and the error probabilities}: 
    Just as in the proof of Lemma \ref{lem:id-error-bound-hayashi}, we start with the observation that the variational distance between output distributions of two different identifiers $i \neq j$ is bounded by
\begin{equation}
    d(P_iW,P_jW) \geq P_iW(\mathcal{D}_i) - P_jW(\mathcal{D}_i) \geq 1-\varepsilon - \delta. \label{eq:watanabe-vd-error-bound}
\end{equation}
    \item \textit{Replacing the input by $M$-types.} We now replace the input $P_i$ for each message $i$ by an $M$-type distribution $\tilde{P}_i$. Lemma \ref{lem:soft-covering-lemma-watanabe} guarantees that we can find an $M$-type distribution such that the variational distance of the partial channel responses is bounded by 
    \begin{equation}
        d(\Tilde{P}_iW^{\mathcal{S}}, P_iW^{\mathcal{S}}) \leq \frac{1}{2} \sqrt{\frac{e^\gamma}{M}}. \label{eq:watanabe-vd-codeword-bound}
    \end{equation}
    \item \textit{Non-distinctness of inputs.} \label{step:proof-watanabe-non-distinct} Notice that the number of distinct $M$-types is upper bounded by $|\mathcal{X}|^M$. However, by assumption, we have $N \geq |\mathcal{X}|^M$, and therefore, there must exist a pair $i'$ and $j'$ such that $\tilde{P}_{i'} = \tilde{P}_{j'}$. For this pair, it follows that also
    \begin{align}
        d(\Tilde{P}_{i'}W^{\mathcal{S}},\Tilde{P}_{j'}W^{\mathcal{S}}) = 0\label{eq:watanabe-vd-output-same-input}
    \end{align}
    must hold since the set $\mathcal{S} = \watanabeset$ is independent of the channel input distribution.
    \item \textit{Estimation of the variational distance with the triangular equality}:
    We will now use \eqref{eq:watanabe-vd-codeword-bound} to find an upper bound on $d(P_{i'}W,P_{j'}W)$ for the pair $(i', j')$ from step \ref{step:proof-watanabe-non-distinct}. Compared to the proof of Lemma \eqref{lem:id-error-bound-hayashi}, we now use the triangular inequality differently to include the restriction on the set $\mathcal{S}$.
 \begin{align}
        \label{eq:d_piw_pjw_leq_sup}
        d(P_{i'}W,P_{j'}W) {\leq}& d(P_{i'}W, \Tilde{P}_{i'}W^{\mathcal{S}}) + \underbrace{d(\Tilde{P}_{i'}W^{\mathcal{S}}, \Tilde{P}_{j'}W^{\mathcal{S}})}_{=0, \text{ since } \Tilde{P}_{i'} = \Tilde{P}_{j'}} + d(\Tilde{P}_{j'}W^{\mathcal{S}}, P_{j'}W) \notag \\
        {\leq}& \underbrace{d(P_{i'}W, P_{i'}W^{\mathcal{S}})}_{= \frac{1}{2} P_{i'}\times W(\mathcal{S}^c) \text{ by } \eqref{eq:var-distance-partial-response-2}} + \underbrace{d(P_{i'}W^{\mathcal{S}}, \Tilde{P}_{i'}W^{\mathcal{S}})}_{\leq\frac{1}{2}\sqrt{\frac{e^\gamma}{M}} \text{ by Lemma } \ref{lem:soft-covering-lemma-watanabe}} \notag\\
        &+ \underbrace{d(\Tilde{P}_{j'}W^{\mathcal{S}}, P_{j'}W^{\mathcal{S}})}_{\leq\frac{1}{2}\sqrt{\frac{e^\gamma}{M}} \text{ by Lemma } \ref{lem:soft-covering-lemma-watanabe}} + \underbrace{d(P_{j'}W^{\mathcal{S}}, P_{j'}W)}_{= \frac{1}{2} P_{j'}\times W(\mathcal{S}^c) \text{ by } \eqref{eq:var-distance-partial-response-2}} \notag \\
        {\leq}& \frac{1}{2}(P_{i'}\times W(\mathcal{S}^c) + P_{j'}\times W(\mathcal{S}^c)) + \sqrt{\frac{e^\gamma}{M}} \notag \\
        \leq& \sup_{P} P\times W(\mathcal{S}^c) + \sqrt{\frac{e^\gamma}{M}}.
    \end{align}
    \item \textit{Completion.} Since \eqref{eq:watanabe-vd-error-bound} holds for every pair $(i,j)$, it must also hold for the specific pair $(i', j')$. Combining \eqref{eq:watanabe-vd-error-bound} and \eqref{eq:d_piw_pjw_leq_sup} leads to
    \begin{align}
        &1-\varepsilon-\delta \leq d(P_iW, P_jW) \leq \sup_{P} P\times W(\mathcal{S}^c) + \sqrt{\frac{e^\gamma}{M}}\notag \\
        \Leftrightarrow&\quad 1 - \sup_{P} P\times W(\mathcal{S}^c) - \sqrt{\frac{e^\gamma}{M}} \leq \varepsilon + \delta \notag \\
        \Leftrightarrow&\quad \inf_{P} P\times W(\mathcal{S}) - \sqrt{\frac{e^\gamma}{M}} \leq \varepsilon + \delta. 
    \end{align}
\end{enumerate}
\end{proof}

\begin{remark}
    The crucial difference to the proof of Lemma \ref{lem:id-error-bound-hayashi} is the way the variational distance is decomposed: while \eqref{eq:hayashi-vd-different-output-prime} makes use of the channel resolvability by splitting the variational distance into 
    \begin{align*}
        d(P_{i'}W,P_{j'}W) &{\leq} \underbrace{d(P_{i'}W,\Tilde{P}_{i'}W)}_{\leq P_{i'}\times W(\mathcal{T}^c) + \frac{1}{2}\sqrt{\frac{e^\gamma}{M}}} + \underbrace{d(\Tilde{P}_{i'}W,\Tilde{P}_{j'}W)}_{=0}+ \underbrace{d(P_{j'}W,\Tilde{P}_{j'}W)}_{\leq P_{j'}\times W(\mathcal{T}^c) + \frac{1}{2}\sqrt{\frac{e^\gamma}{M}}},
    \end{align*}
    in \eqref{eq:d_piw_pjw_leq_sup}, the partial channel resolvability is used:
    \begin{align*}
        d(P_{i'}W,P_{j'}W) {\leq}&  \underbrace{d(P_{i'}W, \Tilde{P}_{i'}W^{\mathcal{S}})}_{\leq \frac{1}{2} P_{i'}\times W(\mathcal{S}^c) + \frac{1}{2}\sqrt{\frac{e^\gamma}{M}}} + \underbrace{d(\Tilde{P}_{i'}W^{\mathcal{S}}, \Tilde{P}_{j'}W^{\mathcal{S}})}_{=0} + \underbrace{d(\Tilde{P}_{j'}W^{\mathcal{S}}, P_{j'}W)}_{\leq \frac{1}{2} P_{j'}\times W(\mathcal{S}^c) + \frac{1}{2}\sqrt{\frac{e^\gamma}{M}}}.
    \end{align*}
    In \eqref{eq:hayashi-vd-different-output-prime}, the resulting upper bound is larger than in \eqref{eq:d_piw_pjw_leq_sup} since the terms $\frac{1}{2} P_{i'}\times W(\mathcal{T}^c)$ and $\frac{1}{2} P_{j'}\times W(\mathcal{T}^c)$ are counted twice, respectively. The reason can easily be seen from the subsequent decompositions: while for \eqref{eq:hayashi-vd-different-output-prime}, we use the decomposition \eqref{eq:proof-scl-hayashi-decompose-vd} in combination with a random coding argument
    \begin{align*}
        \mathbb{E}\left[d(\Tilde{P}W, PW)\right] &{\leq} \underbrace{\mathbb{E}\left[d(\tilde{P}W, \tilde{P}W^{\mathcal{T}})\right]}_{= \frac{1}{2} P\times W(\mathcal{T}^c)}  + \underbrace{\mathbb{E}\left[d(\tilde{P}W^{\mathcal{T}}, PW^{\mathcal{T}})\right]}_{\leq \frac{1}{2}\sqrt{\frac{e^\gamma}{M}}} +  \underbrace{\mathbb{E}\left[d(PW^{\mathcal{T}}, PW)\right]}_{= \frac{1}{2} P\times W(\mathcal{T}^c)},
    \end{align*}
    while in \eqref{eq:d_piw_pjw_leq_sup}, we use
    \begin{align*}
        d(PW, \Tilde{P}W^{\mathcal{S}})\leq& \underbrace{d(PW, PW^{\mathcal{S}})}_{= \frac{1}{2} P\times W(\mathcal{S}^c)} + \underbrace{d(PW^{\mathcal{S}}, \Tilde{P}W^{\mathcal{S}})}_{\leq\frac{1}{2}\sqrt{\frac{e^\gamma}{M}} },
    \end{align*}
    with $P \in \{P_{i'}, P_{j'}\}$. Written this way, it is obvious that the improved bound in Lemma \ref{lem:id-error-bound-watanabe} comes from the fact that we do not need the additional variational distance $d(\tilde{P}W, \tilde{P}W^{\mathcal{T}})$ to get from the partial channel response to the channel response. Note that  partial channel resolvability cannot be used with the set $\mathcal{T}$, since in that case, even though $\Tilde{P}_{i'} = \tilde{P}_{j'}$, \eqref{eq:watanabe-vd-codeword-bound} does not hold:
    \begin{align*}
        d&\left(\Tilde{P}_{i'}W^{\mathcal{T}_{P_i}(\gamma)}, \Tilde{P}_{j'}W^{\mathcal{T}_{P_j}(\gamma)}\right) \\
        &= \frac{1}{2}\sum_{y\in \mathcal{Y}^n} \left|\Tilde{P}_{i'}W(y) \mathds{1}\left[\log\frac{W(y|x)}{P_{i'}W(y)} \leq \gamma \right] - \Tilde{P}_{j'}W(y) \mathds{1}\left[\log\frac{W(y|x)}{P_{j'}W(y)} \leq \gamma \right] \right| \\
        &= \frac{1}{2}\sum_{y\in \mathcal{Y}^n} \left|\Tilde{P}_{i'}W(y) \mathds{1}\left[\log\frac{W(y|x)}{P_{i'}W(y)} \leq \gamma \right] - \Tilde{P}_{i'}W(y) \mathds{1}\left[\log\frac{W(y|x)}{P_{j'}W(y)} \leq \gamma \right] \right|,
    \end{align*}
    which is not necessarily equal to zero since $P_{i'} \neq P_{j'}$. Using the auxiliary output distributions $Q$, this problem is avoided by making the set $\mathcal{S}$ independent of the input distribution.
\end{remark}
\begin{table}[ht!]
    \centering
    \caption{Overview of the similarities and differences in the soft covering lemma by Hayashi and by Watanabe}
    \label{tab:overview_hayashi_watanabe}
    \begin{tabular}{|c|c|c|}
        \hline
        & Hayashi \cite{hayashi} & Watanabe \cite{watanabe}\\
        \hline
        \parbox{1.5cm}{Set} &  \parbox{5cm}{\footnotesize{$\mathcal{T} = \Big\{ (x,y) |\log\frac{W(y|x)}{\color{tu5}PW(y)} \leq \gamma \Big\}$}} & \footnotesize{$ \mathcal{S} = \Big\{ (x,y) | \log| \frac{W(y|x)}{\color{tu5}Q(y)} \leq \gamma \Big\}$}  \\
        \hline
        \parbox{1.5cm}{output\\ distribution} & \parbox{5cm}{true output distribution $PW$} & \parbox{5cm}{auxiliary output distribution $Q$} \\
        \hline
        \multirow{2}*{\parbox{1.5cm}{soft\\ covering\\ lemma}} & \parbox{5cm}{\footnotesize{$d(\Tilde{P}W,PW)\leq \color{tu5}P\!\!\times \!\!W(\mathcal{T}^c)\color{black}+ \frac{1}{2}\sqrt{\frac{e^{\gamma}}{M}}$, any $\gamma\in\mathbb{R}$}} & \parbox{5cm}{\footnotesize{$d(\Tilde{P}W^{\color{tu5}\mathcal{S}},PW^{\color{tu5}\mathcal{S}}) \leq \frac{1}{2}\sqrt{\frac{e^\gamma}{M}}$, any $\gamma\in\mathbb{R}$}} \\
        \cline{2-3}
        & \multicolumn{2}{|c|}{\parbox{10cm}{extra term since Hayashi does not restrict the output distribution on the set $\mathcal{T}$}} \\
        \hline
        \parbox{1.5cm}{$\Tilde{P_{i'}} = \Tilde{P}_{j'} \Rightarrow$} & \parbox{5cm}{$\Tilde{P}_iW^\mathcal{T}=\Tilde{P}_jW^\mathcal{T}$ does not necessarily hold, since $\mathcal{T}$ depends on $P$} & \parbox{5cm}{$\Tilde{P}_iW^\mathcal{S}=\Tilde{P}_jW^\mathcal{S}$ does hold, since $\mathcal{S}$ does not depend on $P$}\\
        \hline
        \multirow{2}*{\parbox{1.5cm}{ID error\\ bound}} & \parbox{5cm}{$\varepsilon+\delta\geq \inf\limits_{P} \color{tu5}2\color{black}PW^{\mathcal{T}} - \sqrt{\frac{e^\gamma}{M}}$} & \parbox{5cm}{$\varepsilon+\delta\geq\inf\limits_{P}PW^{\mathcal{S}}-\sqrt{\frac{e^\gamma}{M}}$} \\
        \cline{2-3}
        & \multicolumn{2}{|c|}{\parbox{10cm}{factor 2 results from the usage of channel resolvability or partial channel resolvability}} \\
        \hline
        \end{tabular}
\end{table}

Table \ref{tab:overview_hayashi_watanabe} provides an overview of the differences and similarities in the soft covering lemma and the resulting bound on the error probabilities of an ID code. The differences are marked in blue: the true/ auxiliary output distribution, the additional term in the soft covering lemma and a factor 2 in the following lemma.

As a result from Lemma \ref{lem:id-error-bound-watanabe} we obtain the following corollary for the maximum code size $N^*$ of an identification code. Note that we use the spectral divergence \eqref{eq:epsilon_spectral_inf_divergence} to express he resulting upper bound.
\begin{corollary}[{\cite[Corollary 1]{watanabe}}]
    \label{cor:loglogN_leq_inf_sup}
    For $0\leq\varepsilon, \delta<1$ with $\varepsilon + \delta < 1$ and an arbitrary $0<\eta<1-\varepsilon-\delta$, we have
    \begin{equation}
        \log\log N^* \leq \inf_{Q} \sup_{P} D_{\mathrm{s}}^{\varepsilon+\delta+\eta} (P\times W\|P\times Q) + \log\log|\mathcal{X}| + 2\log\left(\frac{1}{\eta}\right) + 2.
    \end{equation}
\end{corollary}

\begin{proof}
    For an arbitrary $(N,\varepsilon,\delta)$-ID-code, we choose $M$ such that $N = |\mathcal{X}|^M + 1$ in order for $N>|\mathcal{X}|^M$ to be fulfilled. This means
    \begin{equation}
        M := \Bigl\lfloor \frac{\log(N-1)}{\log|\mathcal{X}|} \Bigr\rfloor \geq \frac{\log(N-1)}{e\log|\mathcal{X}|} \geq \frac{\log N}{e^2\log|\mathcal{X}|} .
    \end{equation}
    Now set
    \begin{equation}
        \label{eq:gamma}
        \gamma := 2\log\eta + \log\log N - \log\log |\mathcal{X}| - 2
    \end{equation}
    such that
    \begin{align}
        \sqrt{\frac{e^\gamma}{M}} &\leq \sqrt{\frac{e^{2\log\eta + \log\log N - \log\log |\mathcal{X}| - 2}}{\frac{\log N}{e^2\log|\mathcal{X}|}}}\notag \\ 
        &= \sqrt{\frac{e^{\log\eta^2}e^{\log\log N}\frac{1}{e^{\log\log|\mathcal{X}|}}\frac{1}{e^2}}{\frac{\log N}{e^2\log|\mathcal{X}|}}} \notag \\
        &= \sqrt{\frac{\eta^2 \cancel{\log N} \frac{1}{\cancel{\log|\mathcal{X}|}} \frac{1}{\cancel{e^2}}}{\frac{\cancel{\log N}}{\cancel{e^2 \log|\mathcal{X}|}}}}\notag \\
        &=\eta. 
    \end{align}
    Since $N>|\mathcal{X}|^M$, Lemma \ref{lem:id-error-bound-watanabe} can be applied and we obtain
    \begin{align}
        \inf_{P} P\times W(\watanabeset) &\leq \varepsilon + \delta + \sqrt{\frac{e^\gamma}{M}} \leq \varepsilon + \delta + \eta \label{eq:proof-loglogn-bound1-infpw}
    \end{align}
    for an arbitrary fixed $Q \in \mathcal{P}(\mathcal{Y})$. On the other hand, by the definition of $D_{\mathrm{s}}^{\varepsilon + \delta + \eta}(P\times W\|P\times Q)$ in \eqref{eq:epsilon_spectral_inf_divergence}, we have
    \begin{align}
        D_{\mathrm{s}}^{\varepsilon + \delta + \eta}(P\times W\|P\times Q) &= \sup\Big\{\Tilde{\gamma} \in \mathbb{R}:\mathbb{P}\Big(\log\frac{W(y|x)}{Q(y)}\leq \Tilde{\gamma} \Big) \leq \varepsilon + \delta + \eta \Big\} \notag\\
        &= \sup\left\{\Tilde{\gamma} \in \mathbb{R}: P\times W\left(\mathcal{S}_Q(\tilde{\gamma})\right) \leq \varepsilon + \delta + \eta \right\}. \label{eq:proof-loglogn-bound1-ds-bound}
    \end{align}
    By \eqref{eq:proof-loglogn-bound1-infpw}, the choice of $\gamma$ in \eqref{eq:gamma} satisfies the condition to be in the set on the right-hand side of \eqref{eq:proof-loglogn-bound1-ds-bound}. Therefore, $\gamma$ is a lower bound for the supremum of that set:
    \begin{align}
        \sup\left\{\Tilde{\gamma} \in \mathbb{R}: P\times W\left(\mathcal{S}_Q(\tilde{\gamma})\right) \leq  \varepsilon + \delta + \eta \right\} &=  D_{\mathrm{s}}^{\varepsilon + \delta + \eta}(P\times W\|P\times Q)\notag\\
        &\geq \gamma \notag\\
        &= 2\log\eta + \log\log N - \log\log |\mathcal{X}| - 2.
    \end{align}
    Solving for $\log \log N$, we find
    \begin{align}
        \log\log N &\leq D_{\mathrm{s}}^{\varepsilon + \delta + \eta}(P\times W\|P\times Q) + 2\log\left(\frac{1}{\eta}\right) + \log\log |\mathcal{X}| +2.
    \end{align}
    The corollary follows by noticing that the bound holds for every $(N, \epsilon, \delta)$-ID code, and therefore also for the maximal code size $N^*$, and for arbitrary $Q\in \mathcal{P}(\mathcal{Y})$.
\end{proof}

By now combining Lemma \ref{lem:divergence_leq_logbeta}, Lemma \ref{lem:saddle_point_property} and Corollary \ref{cor:loglogN_leq_inf_sup} we obtain the minimax bound.

\begin{corollary}[Minimax Bound, {\cite[Corollary 2]{watanabe}}]
    \label{cor:minimax_bound}
    For $0\leq\varepsilon$, $\delta<1$, $\varepsilon+\delta <1$ and an arbitrary $0<\eta<1-\varepsilon-\delta$, we have
    \begin{align}
        \log\log N^* &\leq \min_{Q} \max_{P} -\log\beta_{\varepsilon+\delta+\eta}(P\times W,P \times Q) + \log\log|\mathcal{X}| + 2\log\left(\frac{1}{\eta}\right) + 2 \notag \\
        &= \max_{P} \min_{Q} -\log\beta_{\varepsilon+\delta+\eta}(P \times W,P \times Q) + \log\log|\mathcal{X}| + 2\log\left(\frac{1}{\eta}\right) + 2. 
    \end{align}
\end{corollary}

Now, we are ready to prove the converse coding theorem for identification via channels. 
Consider a sequence $\textbf{X} = (X^n)_{n\in\mathbb{N}}$ and denote by $\textbf{Y} = (Y^n)_{n\in\mathbb{N}}$ the corresponding output sequences and suppose that our ID-code is $(\varepsilon,\delta)$-achievable. \\

\begin{theorem}[{\cite[Theorem 2]{watanabe}}]
    \label{prop:id_capacity_converse}
    For $0\leq\varepsilon$, $\delta<1$, $\varepsilon+\delta<1$ and a sequence $\mathrm{\mathbf{W}} = (W^n)_{n\in\mathbb{N}}$ of general channels, we have
    \begin{equation}
        C_\mathrm{ID}(\varepsilon,\delta|\mathrm{\mathbf{W}}) \leq \sup_{\mathrm{\mathbf{X}}} \underline{I}^{\varepsilon+\delta}(\mathrm{\mathbf{X}};\mathrm{\mathbf{Y}}).
    \end{equation}
\end{theorem}

\begin{proof} 
    By Corollary \ref{cor:minimax_bound} we have
    \begin{equation}
    \label{eq:capcity_proof_1}
        \frac{1}{n} \log\log N_n \leq \max_{P_{X^n}} \min_{Q_{Y^n}} -\frac{1}{n}\log\beta_{\varepsilon_n+\delta_n+\eta_n}(P_{X^n}\times W^n,P_{X^n}\times Q_{Y^n}) + \Delta_n.
    \end{equation}
    with
    \begin{align}
        \Delta_n &:= \frac{1}{n}\left( \log\log|\mathcal{X}^n| + 2\log\left(\frac{1}{\eta_n}\right) + 2 \right) \notag\\
        &= \frac{1}{n}\left(\log n + \log\log|\mathcal{X}| + 2\log\left(\frac{1}{\eta_n}\right) + 2 \right).\label{eq:converse-proof-delta-def}
    \end{align}
    By choosing $\eta = \frac{1}{n}$, we can guarantee that the requirement $\eta < 1-\epsilon-\delta$ is fulfilled for sufficiently large $n$. Now let $\hat{\textbf{X}} = (\hat{X}^n)_{n\in\mathbb{N}}$ and $\hat{\textbf{Y}} = (\hat{Y}^n)_{n\in\mathbb{N}}$ be the input and output sequences that attain the maximum in (\ref{eq:capcity_proof_1}) for every $n$. Therefore, we have
    \begin{align}
    \label{eq:1_n_loglogN_leq_1_n_Ds}
        \frac{1}{n} \log\log N_n &\leq -\frac{1}{n}\log\beta_{\varepsilon_n+\delta_n+\eta_n}(P_{\hat{X}^n}\times W^n,P_{\hat{X}^n}\times Q_{\hat{Y}^n}) + \Delta_n \notag \\
        &\stackrel{\text{Lemma } \ref{lem:divergence_leq_logbeta}}{\leq} \frac{1}{n} D_S^{\varepsilon_n+\delta_n+2\eta_n}(P_{\hat{X}^n}\times W^n\|P_{\hat{X}^n} \times Q_{\hat{Y}^n}) + \frac{1}{n}\log\left(\frac{1}{\eta_n}\right) + \Delta_n, 
    \end{align}
    where the last inequality results from Lemma \ref{lem:divergence_leq_logbeta}.
    Let $\xi = \underline{I}^{\varepsilon+\delta}(\hat{\textbf{X}};\hat{\textbf{Y}}) +  \tau $ for an arbitrary $\tau>0$. We observe that
    \begin{align}
        \mathbb{P}&\left(\frac{1}{n} \log\frac{W^n(\hat{Y}^n|\hat{X}^n)}{P_{\hat{Y}^n}(\hat{Y}^n)} \leq \xi \right) = \mathbb{P}\left(\frac{1}{n} \log\frac{W^n(\hat{Y}^n|\hat{X}^n)}{P_{\hat{Y}^n}(\hat{Y}^n)} \leq \underline{I}^{\varepsilon+\delta}(\hat{\textbf{X}};\hat{\textbf{Y}}) +  \tau \right) \\
         &= \mathbb{P}\Bigg(\color{tu3}\frac{1}{n} \log\frac{W^n(\hat{Y}^n|\hat{X}^n)}{P_{\hat{Y}^n}(\hat{Y}^n)} \color{black} \leq \sup_{a} \left\{\limsup_{n\to\infty} \mathbb{P}\left(\color{tu3}\frac{1}{n} \log\frac{W^n(\hat{Y}^n|\hat{X}^n)}{P_{\hat{Y}^n}(\hat{Y}^n)} \leq a \color{black} \right) \leq \color{tu5}\varepsilon + \delta \color{black} \right\} + \tau \Bigg),\color{black} \label{eq:converse-proof-prob-id-smaller-xi}
    \end{align}
    by the definition of the $\varepsilon$-spectral inf mutual information in \eqref{eq:e-spectral-inf-mutual-information}. Note that on the right-hand side of the inequality in \eqref{eq:converse-proof-prob-id-smaller-xi}, the supremum over $a$ is taken such that the probability $\mathbb{P}\left(\frac{1}{n} \log\frac{W^n(\hat{Y}^n|\hat{X}^n)}{P_{\hat{Y}^n}(\hat{Y}^n)} \leq a  \right)$ is no larger than $\epsilon + \delta$ for large $n$. Denote this supremum by $a^*$. Conversely, that means that for any threshold $a > a^*$, we have $\mathbb{P}\left(\frac{1}{n} \log\frac{W^n(\hat{Y}^n|\hat{X}^n)}{P_{\hat{Y}^n}(\hat{Y}^n)} \leq a  \right) > \varepsilon + \delta$. Since in \eqref{eq:converse-proof-prob-id-smaller-xi}, the threshold is $a^* + \tau$ with $\tau > 0$, we conclude that there exists a $\nu >0$ such that
    \begin{equation}
    \label{eq:pd_xi_e_d_n}
        \color{tu5}\mathbb{P}\left(\frac{1}{n} \log\frac{W^n(\hat{y}^n|\hat{x}^n)}{P_{\hat{Y}^n}(\hat{y}^n)} \leq \xi \right) \geq \varepsilon + \delta + \nu, \color{black}
    \end{equation}
    for infinitely many n. 
   By the assumption that the ID-code is $(\varepsilon,\delta)$-achievable and $\eta_n = \frac{1}{n}$, we know for those values of $n$,
    \begin{align}
        \limsup_{n\to\infty}\varepsilon_n + \delta_n +2\eta_n \leq \varepsilon + \delta.
    \end{align}
    With the definition for the $\varepsilon$-spectral inf-divergence, given in \eqref{eq:epsilon_spectral_inf_divergence}, we obtain
    \begin{align}
    \label{eq:divergence_leq_xi}
        \frac{1}{n} D_S^{\varepsilon_n+\delta_n+2\eta_n}(P_{\hat{X}^n}W^n||P_{\hat{X}^n}Q_{\hat{Y}^n}) &\stackrel{\eqref{eq:epsilon_spectral_inf_divergence}}{=} \sup_\gamma \left\{ \color{tu2}\mathbb{P} \left( \frac{1}{n}\log\frac{W^n(\hat{y}^n|\hat{x}^n)}{P_{\hat{Y}^n}(\hat{y}^n)} \leq \gamma \right) \leq \varepsilon_n+\delta_n+2\eta_n \color{black}\right\}\notag \\ 
        & \leq \color{tu5}\xi\color{black}.
    \end{align}
    Figure \ref{fig:gamma_leq_xi} and the coloring in \eqref{eq:pd_xi_e_d_n} and \eqref{eq:divergence_leq_xi} help to understand the relation in \eqref{eq:divergence_leq_xi}. The blue curve represents the probability density function of $\frac{1}{n}\log\frac{W^n(\hat{Y}^n|\hat{X}^n)}{P_{\hat{Y}^n}(\hat{Y}^n)}$. The value $\xi$ stands for a threshold such that $\mathbb{P}\left(\frac{1}{n} \log\frac{W^n(\hat{Y}^n|\hat{X}^n)}{P_{\hat{Y}^n}(\hat{Y}^n)} \leq \xi  \right) \geq \epsilon + \delta + \nu$, as in \eqref{eq:pd_xi_e_d_n}. 
    If, on the other hand, the probability should be smaller than $\varepsilon+\delta$ (marked in orange), which is obviously less than $\varepsilon+\delta+\nu$ for any $\nu>0$, the largest possible threshold (denoted by $\gamma$ in Fig. \ref{fig:gamma_leq_xi}) is smaller than $\xi$.
    Therefore \eqref{eq:divergence_leq_xi} is fulfilled.
    \begin{figure}[ht]
        \centering
        \includegraphics{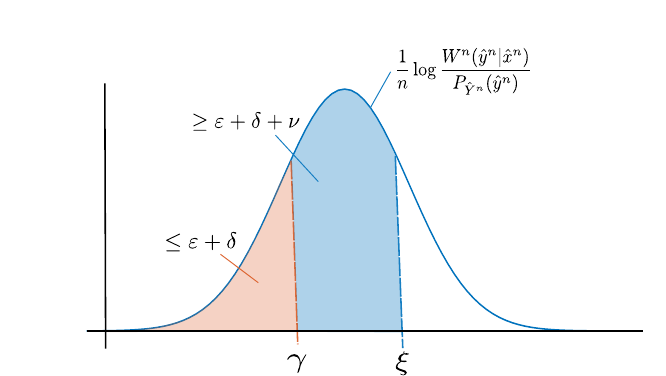}
        \caption{Relations between $\frac{1}{n}\log\frac{W^n(\hat{y}^n|\hat{x}^n)}{P_{\hat{Y}^n}(\hat{y}^n)}$ and its bounds.}
        \label{fig:gamma_leq_xi}
    \end{figure}
    For $\eta_n = \frac{1}{n}$ we have
    \begin{align}
         \liminf_{n\to\infty} \frac{1}{n} \log\log N_n &\leq  \liminf_{n\to\infty} \frac{1}{n} D_S^{\varepsilon_n+\delta_n+2\eta_n}(P_{\hat{X}^n}\times W^n||P_{\hat{X}^n}\times Q_{\hat{Y}^n}) + \frac{1}{n}\log\left(\frac{1}{\eta_n}\right) + \Delta_n \notag \\
        &\stackrel{\eqref{eq:divergence_leq_xi}, \eqref{eq:converse-proof-delta-def}}{\leq}  \liminf_{n\to\infty} \xi + \frac{1}{n} \log n + \frac{1}{n} (\log\log|\mathcal{X}| + 3\log n + 2) \notag \\
        &= \underline{I}^{\varepsilon+\delta}(\hat{\textbf{X}};\hat{\textbf{Y}}) +  \tau \notag\\
        &\leq \sup_{\textbf{X}} \underline{I}^{\varepsilon+\delta}(\textbf{X};\textbf{Y}) +  \tau. \label{eq:converse-proof-finish}
    \end{align}
    Since \eqref{eq:converse-proof-finish} holds for every $\tau > 0$, the claim of the theorem is established.

    For the case $\varepsilon = \delta = 0$, the converse coding theorem follows immediately:
    \begin{corollary}\label{cor:converse-general-watanabe}
         For a sequence $\mathrm{\mathbf{W}} = (W^n)_{n\in\mathbb{N}}$ of general channels, we have
    \begin{equation}
        C_\mathrm{ID}(\mathrm{\mathbf{W}}) = \sup_{\mathrm{\mathbf{X}}} \underline{I}(\mathrm{\mathbf{X}};\mathrm{\mathbf{Y}}).
    \end{equation}
    \end{corollary}
    \mmchange{Comparing Corollary \ref{cor:converse-general-watanabe} to the achievability bound in Theorem \ref{thm:id-capacity-achievable-general}, we find that for $\epsilon = \delta = 0$, the two bounds match, i.e., the identification capacity is established for general (discrete) channels without assumption of the strong converse property.}
\end{proof}
\cdchange{\begin{remark}
    It should be noted that \cite{watanabe} defines general channels with finite input and output alphabets, while \cite{han} defines general channels including continuous input and output alphabets. The corollary~\ref{cor:converse-general-watanabe} therefore does not apply to continuous input or output alphabets. The strong converse property must be shown for such channels. For example, this applies to the deterministic-time Poisson channel with power constraint \cite{labidi24}.
\end{remark}}

\section{Converse for Identification in the Presence of Feedback}
\label{idf}
In this section, we analyze the converse proof of the ID capacity in the presence noiseless feedback \mmchange{by Ahlswede and Dueck \cite{ahlswede_feedback}. Since that proof is inspired by the strong converse proof for transmission with feedback and bears many similarities with it, in Section \ref{sec:former_results_on_transmission_with_feedback}, we first introduce the converse for transmission with feedback as it is presented by Wolfowitz \cite[Theorem 4.8.2]{wolfowitz}. He attributes that proof to an oral communication with H. Kesten and J. H. Kemperman. Subsequently, in Section \ref{sec:converses_idf}, we present the converse proof for identification with feedback \cite{ahlswede_feedback}. Finally, in Section \ref{sec:comparison-tf-idf}, we highlight} the similarities and differences between these two approaches.

\subsection{Converse for Transmission with Feedback}
\label{sec:former_results_on_transmission_with_feedback}
In this section, we introduce the converse proof of transmission with feedback through a discrete memoryless channel (DMC).\mmchange{
\begin{theorem}[Converse for transmission with feedback {\cite[Theorem 4.8.2]{wolfowitz}}]
    Any $(n, M, \varepsilon)$-transmission feedback code for a DMC $W$ with feedback satisfies
    \begin{align}
        \log M < nC(W) + K\sqrt{n},
    \end{align}
    where $C(W)$ denotes the Shannon capacity from Theorem \ref{thm:shannon} and $K$ depends on $\epsilon$ but not on $n$.
\end{theorem}}
\begin{proof}
The proof by Wolfowitz \mmchange{is given in \cite[p. 95]{wolfowitz} and} can be structured in five steps. 
\begin{enumerate}
\item Let $t\in \{2,\ldots,n\}$ and the auxiliary RV $Y_t^*$ be defined as
\begin{equation}
\label{eq:wolfowitz_def_yt_star}
    Y_t^* = \log\frac{W(Y_t|f_{i}^t(Y^{t-1}))}{P^*_Y(Y_t)},
\end{equation} 
where $Y^{t-1}=(Y_1,\cdots,Y_{t-1})$ comprises all channel outputs until time $t-1$ and $P^*_Y$ represents the output distribution of the DMC $W$ that corresponds to the optimal input distribution $P^*:=\arg\max_{P}I(P;W)$, i.e.,  
\begin{equation*}
    P^*_Y(y)=\sum_{x\in\mathcal{X}}P^*(x)W(y|x), \quad \forall y\in\mathcal{Y}.
\end{equation*}

For all $y^{t-1}\in\mathcal{Y}^{t-1}$, we have
\begin{align}
\label{eq:wolfowitz_expectation_yt}
    \mathbb{E}[Y_t^*|y^{t-1}] &= \sum\limits_{y_t \in\mathcal{Y}} Y_t^*(y_t) \mathbb{P}(y_t|y^{t-1}) \notag \\
    &= \sum\limits_{y_t \in\mathcal{Y}} \log \Big(\frac{W(y_t|f_{t}^t(y^{t-1}))}{P^*_Y(y_t)}\Big) \cdot W(y_t|f_{i}^t(y^{t-1})) \notag \\
    &= D(W(\cdot|f_i^t(y^{t-1}))||P_{Y^*_t})\notag \\
    &\stackrel{(a)}{\leq} C,
\end{align}
where $(a)$ follows from \textit{Shannon's lemma}
\renewcommand*{\marginnotevadjust}{3cm} (see \cite[Theorem 65]{ahlswede_book_2}, \cite[Lemma 4.25]{alajaji} and \cite[p. 90]{wolfowitz}). We denote the argument that \eqref{eq:wolfowitz_expectation_yt} holds true for all $y^{t-1}\in\mathcal{Y}^{t-1}$ by $\mathbb{E}[Y^*_t|Y^{t-1}]\le C$. 
\item \label{step2_Wolofitz}
Let the auxiliary RV $U_t$ be defined as
\begin{equation}
\label{eq:wolfowitz_def_ut}
    U_t = Y_t^* - \mathbb{E}[Y_t^*|Y^{t-1}]. 
\end{equation}
The RV $U_t$ fulfills the following properties:
\begin{itemize}
    \item[(a)] $\mathbb{E}[U_t|Y^{t-1}] = 0$,
    \item[(b)] For all $s<t$, $\mathbb{E}[U_t|U_s] = 0$,
    \item[(c)] For all $s\ne t$, $U_t$ and $U_s$ are uncorrelated.
\end{itemize}
\textit{Proof:}
\begin{itemize}
    \item[(a)]
    \begin{align}
    \label{eq:1}
        \mathbb{E}[U_t|Y^{t-1}]
        &=\mathbb{E}[Y_t^*-\mathbb{E}[Y_t^*|Y^{t-1}]|Y_1,...,Y_{t-1}] \nonumber\\
        &=\mathbb{E}[Y_t^*|Y^{t-1}]-\mathbb{E}[\mathbb{E}[Y_t^*|Y^{t-1}]|Y^{t-1}]\nonumber\\
        &= 0.
    \end{align}
    \item[(b)]
    \begin{align}
        \mathbb{E}[U_t|U_s] = \mathbb{E}[U_t|Y_s] \stackrel{\text{(a)}}{=} 0.
    \end{align}
    \item[(c)]
    \begin{align}
        \frac{\mathbb{E}[U_t U_s]}{P(U_s)} = \mathbb{E}[U_t|U_s] \stackrel{\text{(b)}}{=} 0.
    \end{align}
\end{itemize}
For some $\alpha>0$, consider the event 
\begin{equation}
    \sum\limits_{t=1}^{n} U_t >\alpha\sqrt{n}. \label{eq:RVU_t}
\end{equation}
 By the definition of $U_t$ and the bound for $\mathbb{E}[Y_t^*|Y^{t-1}]$ in \eqref{eq:wolfowitz_expectation_yt}, the event in \eqref{eq:RVU_t} is equivalent to the following event
\begin{align}
    \sum\limits_{t=1}^{n} Y_t^*\nonumber
    &>\alpha\sqrt{n} + \sum\limits_{t=1}^{n}\mathbb{E}[Y_t^*|Y^{t-1}].
\end{align}
By \eqref{eq:wolfowitz_expectation_yt} we have $\alpha\sqrt{n} + \sum\limits_{t=1}^{n}\mathbb{E}[Y_t^*|Y^{t-1}]\le \alpha\sqrt{n} +nC$. Therefore,
\begin{align}
\label{eq:wolfowitz_probs_coloured}
    \mathbb{P}\Big(\sum\limits_{t=1}^{n} U_t > \alpha\sqrt{n}\Big)
    &=\color{black}\mathbb{P}\Big(\sum\limits_{t=1}^{n} Y_t^* >\alpha\sqrt{n} + \sum\limits_{t=1}^{n}\mathbb{E}[Y_t^*|Y^{t-1}]\Big)\nonumber\\
   &\ge \color{black}\mathbb{P}\Big(\sum\limits_{t=1}^{n} Y_t^* >\alpha\sqrt{n} + nC\Big). 
\end{align}
 The analyze the probability $\mathbb{P}\left(\sum_{t=1}^{n} U_t\ge\alpha\sqrt{n}\right)$ further using the Chebyshev inequality
\begin{equation}
    \label{eq:chebyshev}
    \mathbb{P}(|X-\mu|\geq k\sigma) \leq \frac{1}{k^2},
\end{equation}
with the following parameters:
\begin{itemize}
    \item[(a)] RV:
    \begin{align}
        X=\sum_{t=1}^{n}U_t.
    \end{align}
    \item[(b)] Mean:
    \begin{align}
        \mu=\mathbb{E}[\sum_{t=1}^{n}U_t]=\sum_{t=1}^{n}\mathbb{E}[\mathbb{E}[U_t|Y^{t-1}]]\stackrel{\eqref{eq:1}}{=}0.
    \end{align}
    \item[(c)] Variance:
    \begin{align}
        \sigma^2=\text{Var}[\sum_{t=1}^n U_t] 
        &= \sum_{t=1}^n\mathbb{E}[U_t^2]+\sum_{1\le s<t\le n} \text{Cov}[U_s,U_t]
        &\stackrel{\eqref{eq:wolfowitz_def_ut}\text{(c)}}{=} \sum_{t=1}^n\mathbb{E}[U_t^2].
    \end{align}
    Define $\beta$ as an upper-bound of $\mathbb{E}[U_t^2]$ for all $t=1,\cdots,n$, i.e., $\sigma^2=\sum_{t=1}^n\mathbb{E}[U_t^2]\le n\beta$.
    \item[(d)] Parameter k:
    \begin{align}
        k := \frac{\alpha}{\sqrt{\beta}}=\frac{1}{\sqrt{\nu}},
    \end{align}
    where $\nu>0$ and$\alpha$ is defined in step 2.2.
\end{itemize} 
Applying Chebyshev's inequality, we have
\begin{equation}
\label{eq:wolfowitz_prob_bounded_by_chebyshev}
    \mathbb{P}(\sum\limits_{t=1}^{n} U_t \geq \alpha\sqrt{n}) \leq \nu.
\end{equation}
Note that $\beta$ can be any upper bound of $\mathbb{E}[U^2_t]$. We will see later that Ahlswede provides a concrete value for $\beta$.
By combining \eqref{eq:wolfowitz_probs_coloured} and \eqref{eq:wolfowitz_prob_bounded_by_chebyshev}, we obtain
\begin{align}
\label{eq:wolfowitz_prob_leq_nu}
    \mathbb{P}\Big(\sum\limits_{t=1}^{n} Y_t^* >\alpha\sqrt{n} + nC\Big) & \leq \mathbb{P}\Big(\sum\limits_{t=1}^{n} U_t > \alpha\sqrt{n}\Big) \nonumber \\
    & \leq \nu.
\end{align}

\item Next, for all $i\in\mathcal{M}$, we define the set $\mathcal{C}(P_Y) \subset \mathcal{Y}^n $ as follows.
\begin{equation}
\label{eq:c_p_y_subset}
    \mathcal{C}(P_Y) =\Big\{y^n\Big|\log\frac{W^n(y^n|\boldsymbol{f}_i)}{P_Y(y^n)} \geq nC + \alpha\sqrt{n} =: \theta \Big\} .
\end{equation}
Then, for all $i\in\mathcal{M}$, we have
\begin{align}
    W^n(\mathcal{C}(P^*_Y)|\boldsymbol{f}_i) &= \mathbb{P}\Big(\log\prod\limits_{t=1}^{n} \frac{W(Y_t|f^t_{i}(Y^{t-1}))}{P^*_Y(Y_t)} >\alpha\sqrt{n} + nC\Big) \notag \\
    &=\mathbb{P}\Big(\sum\limits_{t=1}^{n} Y_t^* >\alpha\sqrt{n} + nC\Big)\nonumber\\
    &\stackrel{\eqref{eq:wolfowitz_prob_leq_nu}}{\leq} \nu.
\end{align}

\begin{lemma}[{\cite[Theorem 7.8.1]{wolfowitz}}]
\label{lem:wolfowitz}
For $\nu>0$, suppose 
\begin{equation}
    \min\limits_{P_Y} \max\limits_{f}  W^n(\mathcal{C}(P_Y)|f) < \nu,
\end{equation} 
Then
\begin{equation}
    M < \frac{2^\theta}{1-\nu-\varepsilon},
\end{equation}
where $\theta$ is defined in \eqref{eq:c_p_y_subset} and $\varepsilon$ is the upper bound of error probability defined in \eqref{eq:error}.
\end{lemma}

We can upper-bound $W^n(\mathcal{D}_i\cap\mathcal{C}(P_Y^*)^c|f_i)$ as follows.
\begin{align}
\label{eq:2}
    W^n(\mathcal{D}_i\cap\mathcal{C}(P_{Y^*})^c|f_i) &= W^n(\mathcal{D}_i|f_i)-W^n(\mathcal{D}_i\cap\mathcal{C}(P_Y^*)|f_i) \notag \\
    &\geq\underbrace{W^n(\mathcal{D}_i|f_i)}_{>1-\varepsilon} - \underbrace{W^n(\mathcal{C}(P^*_{Y})|f_i)}_{<\nu } \notag \\
    &> 1-\varepsilon-\nu.
\end{align}
\begin{proof}
    By the definition of the subset $\mathcal{C}(P_Y)$ in \eqref{eq:c_p_y_subset} we obtain
    \begin{equation}
    \label{eq:wolfowitz_proof_1}
        2^\theta \cdot P_Y^*(\mathcal{D}_i\cap\mathcal{C}(P_Y^*)^c) \geq W^n(\mathcal{D}_i\cap\mathcal{C}(P_Y^*)^c|f_i).
    \end{equation}
    Substitute \eqref{eq:2} into \eqref{eq:wolfowitz_proof_1} and take the sum over $i$ from $1$ to $M$, we have
    \begin{align}
        2^\theta &\geq 2^\theta\underbrace{\sum\limits_{i=1}^{M} P_{Y}^*(\mathcal{D}_i\cap \mathcal{C}(P_Y^*))}_{\leq1} \stackrel{\eqref{eq:wolfowitz_proof_1}}{\geq} 2^{\theta}\sum\limits_{i=1}^{M} W^n(\mathcal{D}_i\cap\mathcal{C}(P_Y^*)^c|f_i) \notag\\
        &\stackrel{\eqref{eq:2}}{>} \sum\limits_{i=1}^{M} 1-\nu-\varepsilon = M(1-\nu-\varepsilon).
    \end{align}
\end{proof}

\item By applying Lemma \ref{lem:wolfowitz}, we obtain 
\begin{align}
    M < \frac{2^\theta}{1-\nu-\varepsilon} &\Leftrightarrow \log M < \theta-\log(1-\nu-\varepsilon) \notag \\
    & \quad \quad \quad \quad \stackrel{\eqref{eq:c_p_y_subset}}{=} nC +\alpha\sqrt{n}-\log(1-\nu-\varepsilon) \notag \\
    &\Leftrightarrow \frac{1}{n}\log M < C + \underbrace{\frac{\alpha}{\sqrt{n}}}_{\xrightarrow[n\to\infty]{} 0} + \underbrace{\frac{1}{n}\log(1-\nu-\varepsilon)}_{\xrightarrow[n\to\infty]{} 0}.
\end{align}
Finally, we have
\begin{equation}
    R:=\limsup_{n \to \infty} \frac{1}{n} \log M < C.
\end{equation}
\end{enumerate}
\end{proof}


\subsection{Converse Proof for Identification with Feedback}
\label{sec:converses_idf}
In this section, we examine the converse proof for the deterministic ID feedback capacity established by Ahlswede and Dueck in \cite{ahlswede_feedback}. We first highlight the similarities to Wolfowitz's proof \cite{wolfowitz}. We then revisit the converse proof for the randomized ID feedback capacity \cite{ahlswede_feedback}.

By introducing a RV $Y_t^*$ \mmchange{in \eqref{eq:wolfowitz_def_yt_star}} that contains the channel probability $W(Y_t|f_{i}^t(Y^{t-1}))$, we account for the feedback strategy $f_i$ and on the known variables $Y^{t-1}$. This feedback is incorporated through $Y_t^*$, whose expectation, conditioned on $Y^{t-1}$, can be upper bounded by the channel capacity \mmchange{\eqref{eq:wolfowitz_expectation_yt}}. According to Lemma \ref{lem:wolfowitz} \cite[Theorem 7.8.1]{wolfowitz}, there exists a transmission feedback code with the number of messages growing exponentially with the number of bits and the Shannon capacity.
 
Ahlswede and Dueck extended this idea of using conditional expectation for feedback in their 1989 publication on identification codes. The key difference from Wolfowitz's work lies in the capacity formula and the double exponential growth in the number of messages. The proofs differ primarily in the step where double exponential growth is derived from an auxiliary statement\footnote{Ahlswede: Lemma \ref{lem:idf_subset_n} \cite[Lemma 3]{ahlswede_feedback}, Wolfowitz: Lemma \ref{lem:wolfowitz} \cite[Theorem 7.8.1]{wolfowitz}}.

\subsubsection{Deterministic ID with Feedback}
\mmchange{\begin{theorem}[Converse for deterministic identification with feedback {\cite[Theorem 1, part b)]{ahlswede_feedback}}]
    If the transmission capacity $C(W)$ is positive, then we have for all $\lambda \in 0, \frac{1}{2}$:
    \begin{align}
        \lim_{n\to\infty} \inf \frac{1}{n} \log \log N_f(n, \lambda) \leq \max_{x \in \mathcal{X}} H(W(\cdot|x)).
    \end{align}
\end{theorem}}
\begin{proof}
The proof is structured analogously to Wolfowitz' proof.
\begin{enumerate}
\item Let $t \in \{2,\ldots,n\}$ and let the RV $ Y_t^*$ be defined as the following:
\begin{equation}
\label{eq:def_yt_star}
    Y_t^* = -\log W(Y_t|f^t_{i}(Y^{t-1})).
\end{equation}
For all $y^{t-1}\in\mathcal{Y}^{t-1}$, we have
\begin{align}
\label{eq:expectation_leq_entropy_channel}
    \mathbb{E}[Y_t^*|y^{t-1}] \leq H(W(\cdot | x^*)), 
\end{align}
where $H(W(\cdot | x^*)) = \max\limits_{x\in\mathcal{X}} H(W(\cdot|x))$.\newline
\textit{Proof:}
\begin{align}
    \mathbb{E}[Y_t^*|y^{t-1}] &= \sum\limits_{y_t \in\mathcal{Y}} Y_t^*(y_t) \mathbb{P}(y_t|y^{t-1}) \notag \\
    &= -\sum\limits_{y_t \in\mathcal{Y}} \log W(y_t|f^t_{i}(y^{t-1})) \cdot W(y_t|f^t_{i}(y^{t-1})) \notag \\
    &= H(W(\cdot|f^t_{i}(y^{t-1})) \notag \\
    &\leq H(W(\cdot|x^*)).
\end{align}
Therefore, for all $y^{t-1}\in\mathcal{Y}^{t-1}$, we can upper-bound each realization of $\mathbb{E}[Y_t^*|Y^{t-1}]$ by $H(\cdot|x^*)$, and we denote this result as $\mathbb{E}[Y_t^*|Y^{t-1}]\le H(\cdot|x^*)$.

\item Let the RV $U_t$ be defined as the following:
\begin{equation}
    \label{eq:def_ut}
    U_t = Y_t^* - \mathbb{E}[Y_t^*|Y^{t-1}].
\end{equation}
Ahlswede and Dueck also proved that the RV $U_t$ fulfills the same properties described in Step \ref{step2_Wolofitz} of Wolfowitz's proof 
Let us consider the event $\sum_{t=1}^{n}U_t\le\alpha\sqrt{n}$.
If $\alpha>0$, we have
\begin{equation}
    \sum\limits_{t=1}^{n} U_t = \sum\limits_{t=1}^{n} Y_t^* - \sum\limits_{t=1}^{n} \mathbb{E}[Y_t^*|Y_1,...,Y_{t-1}] \leq \alpha \sqrt{n} .
\end{equation}
Thus, we have 
\begin{equation}
    \sum\limits_{t=1}^{n} Y_t^* \leq \alpha \sqrt{n} + \sum\limits_{t=1}^{n} \underbrace{\mathbb{E}[Y_t^*|Y_1,...,Y_{t-1}]}_{\leq H(W(\cdot|x^*))} \leq \alpha \sqrt{n} + n H(W(\cdot|x^*)). 
\end{equation}
Consequently, we obtain
\begin{align}
\label{eq:1_minus_pr_ut}
     \color{black}\mathbb{P}\Big(\sum\limits_{t=1}^{n} Y_t^* \leq \alpha \sqrt{n} + n &\color{black}H(W(\cdot|x^*))\Big)\color{black}\notag \\
     &\geq \color{black}\mathbb{P}\Big(\sum\limits_{t=1}^{n} Y_t^* \leq \alpha\sqrt{n}+\sum\limits_{t=1}^{n} \mathbb{E}[Y_t^*|Y_1,...,Y_{t-1}]\Big) \color{black}\notag \\
    &= \mathbb{P}\Big(\sum\limits_{t=1}^{n} U_t \leq \alpha \sqrt{n}\Big) \notag \\ 
    &= 1 - \mathbb{P}\Big(\sum\limits_{t=1}^{n} U_t > \alpha \sqrt{n}\Big). 
\end{align}

 Wolfowitz chose an arbitrary $\beta$ as an upper bound for the variance, whereas Ahlswede identified a specific $\beta$:
\begin{align}
    \beta=\max\left\{\log^2{3},\log^2{|\mathcal{Y}|}\right\}.
\end{align}
By Applying Chebyshev's inequality, we obtain
\begin{align}
\label{eq:chebyshev_pr_ut}
    \mathbb{P}\Big(\sum\limits_{t=1}^{n} U_t \geq \alpha \sqrt{n}\Big) 
    & \leq \nu. 
\end{align}

By combining \eqref{eq:1_minus_pr_ut} and \eqref{eq:chebyshev_pr_ut}, we obtain
\begin{align}
\label{eq:prob_geq_1_minus_nu}
    \mathbb{P}\Big(\sum\limits_{t=1}^{n} Y_t^* \leq \alpha \sqrt{n} + n H(W(\cdot|x^*))\Big) &\geq 1 - \mathbb{P}\Big(\sum\limits_{t=1}^{n} U_t > \alpha \sqrt{n}\Big)\notag \\
    &\geq 1-\nu.
\end{align}
From the definition of $Y_t^*$ in \eqref{eq:def_yt_star}, we have
\begin{align}
\label{eq:prob_yt_logw}
    \mathbb{P}\Big(\sum\limits_{t=1}^{n} Y_t^* \leq \alpha \sqrt{n} + n &H(W(\cdot|x^*))\Big) \notag \\
    &= \mathbb{P}\Big(\sum\limits_{t=1}^{n} -\log W(Y_t|\boldsymbol{f}_i^t(Y^{t-1})) \leq \alpha \sqrt{n} + n H(W(\cdot|x^*))\Big) \notag \\
    &= \mathbb{P}\Big(-\log \prod\limits_{t=1}^{n} W(Y_t|\boldsymbol{f}_i^t(Y^{t-1})) \leq \alpha \sqrt{n} + n H(W(\cdot|x^*))\Big). 
\end{align}
By combining \eqref{eq:prob_geq_1_minus_nu} and \eqref{eq:prob_yt_logw}, we conclude
\begin{equation}
\label{eq:pr_log_prod_geq_1_minus_nu}
    \mathbb{P}\Big(-\log \prod\limits_{t=1}^{n} W(Y_t|\boldsymbol{f}_i^t(Y^{t-1})) \leq \alpha \sqrt{n} + n H(W(\cdot|x^*))\Big) \geq 1 - \nu.
\end{equation}

\item Let $\mathcal{E}$ be given by
\begin{equation}
    \label{eq:set_e}
    \mathcal{E} = \Big\{y^n\in\mathcal{Y}^n|-\log W^n(y^n|f) \leq \alpha\sqrt{n} + nH(W(\cdot|x^*)) := \log K \Big\}.
\end{equation}
In other words, the set $\mathcal{E}$ contains all sequences $y^n \in \mathcal{Y}^n$ satisfying
\begin{align}
\label{eq:3}
    W^n(y^n|\boldsymbol{f}) \geq \frac{1}{K}.
\end{align}
We have
\begin{align}
    1&=\sum_{y^n\in\mathcal{Y}^n} W^n(y^n|\boldsymbol{f})\nonumber\\
    &\ge \sum_{y^n\in\mathcal{E}}W^n(y^n|\boldsymbol{f})\nonumber\\
    &\overset{(a)}{\ge}|\mathcal{E}|\cdot\frac{1}{K}, 
\end{align}
where $(a)$ follows from \eqref{eq:3}.
Thus, we have $|\mathcal{E}|\le K$.
From the definition of $\mathcal{E}$ in \eqref{eq:set_e}, we can rewrite \eqref{eq:pr_log_prod_geq_1_minus_nu} as 
\begin{equation}
    \label{eq:wn_geq_1_minus_nu}
    \mathbb{P}\Big(-\log \prod\limits_{t=1}^{n} W(Y_t|f) \leq \alpha \sqrt{n} + n H(W(\cdot|x^*))\Big) = W^n(\mathcal{E}|f)\geq 1 - \nu.
\end{equation}

\item The following lemma was demonstrated by Ahlswede and Dueck in \cite{ahlswede_feedback}.
\begin{lemma}[\cite{ahlswede_feedback}]
\label{lem:idf_subset_n}
    For any feedback strategy $f$ and $\nu \in (0,1)$ with
    \begin{equation}
        \min\limits_{\mathcal{E}\subset\mathcal{Y}^n:W^n(\mathcal{E}|f)\geq 1-\nu} |\mathcal{E}| \leq K =: 2^{nH(W(\cdot|x^*))+\alpha\sqrt{n}}, 
    \end{equation}
    where $\alpha = \sqrt{\frac{\beta}{\nu}}$ and $\beta = |\mathcal{Y}|\log^23$, we have
    \begin{equation}
        N< 2^{n\log |\mathcal{Y}| K}. 
    \end{equation}
\end{lemma}

\begin{proof}
The number of messages $N$ can be upper bounded by
    \begin{equation}
        N\leq \underbrace{\sum\limits_{i=0}^{K} \underbrace{{|\mathcal{Y}|^n}\choose{i}}_{\begin{tabular}{c} \text{possible subsets}\\ $\mathcal{E}\subset\mathcal{Y}^n$ \text{ of size }$i$\end{tabular}}}_{\begin{tabular}{c} \text{all possible subsets }$\mathcal{E}\subset\mathcal{Y}^n$ \\ \text{of all possible sizes }$0,...,K$\end{tabular}}
        {\leq} (|\mathcal{Y}|^n)^K = (2^{\log |\mathcal{Y}|^n})^K = 2^{\log |\mathcal{Y}|^nK} = 2^{n\log |\mathcal{Y}| K}, 
    \end{equation}
    where the first inequality follows from the fact that \( \mathcal{D}_i \cap \mathcal{E}_i \) are distinct for \( i=1, \ldots, N \) (implying the existence of \( N \) subsets \( \mathcal{E}_i \). In other words, \( N = |\bigcup_{i=1}^{n} \mathcal{E}_i| \)).
 
At this point, the proofs diverge significantly. Ahlswede and Wolfowitz both use auxiliary statements (Lemma \ref{lem:wolfowitz} and Lemma \ref{lem:idf_subset_n}, respectively) to show exponential or double exponential growth of the identification rate, bounding the number of messages \( N \) or \( M \) under specific conditions.

In Lemma \ref{lem:wolfowitz}, this condition requires \( \min_{P_Y} \max_{f} W^n(\mathcal{C}(P_Y) \mid f) < \nu \), allowing \( M \) to be bounded by a threshold \( \theta \) related to \( \mathcal{C}(P_Y) \). In contrast, for Lemma \ref{lem:idf_subset_n}, Ahlswede specifies that the cardinality of the subset \( \mathcal{E} \) must be less than a constant \( K \). Similar to Wolfowitz, the number of messages \( N \) can then be bounded by an expression involving the threshold \( K \) from the subset \( \mathcal{E} \).

In both conditions, the channel probability with elements from the subsets $\mathcal{C}(P_Y)$ or $\mathcal{E}$ must be bounded (highlighted in purple). The key difference is Ahlswede's focus on the size of the subset. Additionally, the proofs for bounding $N$ and $M$ differ. Ahlswede uses the distinctness of the subsets $\mathcal{D}_i \cap \mathcal{E}_i$, while Wolfowitz relies solely on the definition of the set $\mathcal{C}(P_Y)$.

We can then bound  $W^n(\mathcal{D}_i \cap \mathcal{E}_i|f_i)$ using lemma \ref{lem:idf_subset_n}.

The distinctness of $\mathcal{D}_i \cap \mathcal{E}_i$ for $i = 1, \ldots, N$ can be proven by contradiction. This requires finding a bound for $W^n(\mathcal{D}_i \cap \mathcal{E}_i|f_i)$. \\
    \begin{addmargin}[25pt]{0pt}
    First notice that
    \begin{align}
        W^n(\mathcal{D}_i\cap \mathcal{E}_i|f_i) &{=} W^n(\mathcal{E}_i|f_i) - W^n(\mathcal{E}_i|\mathcal{D}_i|f_i) \notag \\
        &{\geq} \underbrace{W^n(\mathcal{E}_i|f_i)}_{\geq 1-\nu}  - \underbrace{W^n(\mathcal{D}_i^c|f_i)}_{\leq \varepsilon} \notag \\
        &\geq 1-\nu-\varepsilon, 
    \end{align}
    where the first equality and inequality follow from probabilistic properties akin to those in Wolfowitz's proof. Here, \( W^n(\mathcal{E}_i \mid f_i) > 1 - \nu \) (from \eqref{eq:wn_geq_1_minus_nu}), and \( W^n(\mathcal{D}_i^c \mid f_i) < \varepsilon \) denotes the type I error probability.

Suppose \( \mathcal{D}_i \cap \mathcal{E}_i = \mathcal{D}_j \cap \mathcal{E}_j \) for \( i \neq j \). Then \( W^n(\mathcal{D}_j \cap \mathcal{E}_j \mid f_i) < \delta \) signifies the type II error probability.

    \begin{equation}
        1-\nu-\varepsilon \leq W^n(\mathcal{D}_i\cap \mathcal{E}_i|f_i) = \underbrace{W^n(\mathcal{D}_j\cap \mathcal{E}_j|f_i)}_{\text{type II error}} < \delta \Leftrightarrow 1-\nu < \varepsilon + \delta \quad \lightning \varepsilon + \delta < 1 
    \end{equation}
    must hold, which contradicts the fact, that the sum of both error probabilities must be smaller than one, for any $\nu>0$. 
    \end{addmargin}
\end{proof}

Note that in this step Ahlswede uses the subset $\mathcal{E}$ in the argument of $W^n(\cdot|f)$, while Wolfowitz uses the complementary subset $\mathcal{C}(P_Y^*)^c$, which results from different inequality sign in former steps.
Besides that Ahlswede uses the distinctness of $\mathcal{D}_i\cap\mathcal{E}_i$ to prove Lemma \ref{lem:idf_subset_n}, while Wolfowitz works with the definition of the subset $\mathcal{C}(P_Y)$. However, since both can bound the channel probabilities by $W^n(\mathcal{C}(P_Y^*)|f) < \nu$ or $W^n(\mathcal{E}|f)\geq 1-\nu$ respectively, Lemma \ref{lem:wolfowitz} or Lemma \ref{lem:idf_subset_n} is applicable and we can conclude the growth of the coding rates by the respective capacity in the following step.

\item 
By using Lemma \ref{lem:idf_subset_n} we obtain
\begin{align}
    N < 2^{n\log |\mathcal{Y}|K} &\Leftrightarrow \log N < n\log |\mathcal{Y}| K \notag \\
    &\Leftrightarrow \log N < n\log |\mathcal{Y}| 2^{nH(W(\cdot|x^*))+\alpha\sqrt{n}} \notag \\
    &\Leftrightarrow \log\log N < \log (n\log |\mathcal{Y}|) + nH(W(\cdot|x^*))+\alpha\sqrt{n} \notag \\
    &\Leftrightarrow \frac{1}{n}\log\log N <\frac{1}{n}\log (n\log |\mathcal{Y}|)+ H(W(\cdot|x^*))+ \frac{\alpha}{\sqrt{n}}.
\end{align}
Taking the limes superior an both sides, he finally gets
\begin{equation}
    \limsup_{n \to \infty} \frac{1}{n} \log\log N < H(W(\cdot|x^*)).
\end{equation}

\end{enumerate}
\end{proof}

\subsubsection{Randomized ID with Feedback}
For completeness we will provide the converse for randomized IDF-codes in the following.
\mmchange{\begin{theorem}[Converse for radomized identification with feedback {\cite[Theorem 2, part b)]{ahlswede_feedback}}]
    If the transmission capacity $C(W)$ is positive, then we have for all $\lambda \in 0, \frac{1}{2}$:
    \begin{align}
        \lim_{n\to\infty} \inf \frac{1}{n} \log \log N_F(n, \lambda) \leq \max_{P \in \mathcal{P}(\mathcal{X})} H(P\cdot W).
    \end{align}
\end{theorem}}

The proof for a randomized encoder is very similar to the proof for the deterministic encoder. We simply define the RV's different, labeled by an apostrophe $'$ in the notations, but use the same methods.

Define the RV
\begin{align}
    Y_t' = -\log P_Y(Y_t|f^t_{i}(Y_1,...,Y_{t-1})) 
\end{align}
for
\begin{equation}
    P_Y(Y_t|f^t_{i}(Y_1,...,Y_{t-1})) = Q_f(f^t_{i}(Y_1,...,Y_{t-1}))W(Y_t|f^t_{i}(Y_1,...,Y_{t-1})) 
\end{equation}
with the property
\begin{equation}
    \mathbb{E}[Y_t'|Y_1,...,Y_{t-1}] \leq \max\limits_{P} H(P W) = H(P^* W), 
\end{equation}
where $P^*$ denotes the input distribution that maximizes the output entropy, 
since
\begin{align}
\mathbb{E}[Y_t'|Y_1,...,Y_{t-1}] &= \sum\limits_{Y_t\in \mathcal{Y}} Y_t'\mathbb{P}(Y_t|Y_1,...,Y_{t-1}) \notag \\
&= -\sum\limits_{Y_t\in \mathcal{Y}} \log P_Y(Y_t|f^t_{i}(Y_1,...,Y_{t-1})) \cdot P_Y(Y_t|f^t_{i}(Y_1,...,Y_{t-1})) \notag \\
&= H(P_Y|P_{Y_1},...,P_{Y_{t-1}}) \notag \\
&\leq H(P_Y) = H(P W) \notag \\
&\leq \max\limits_{P} H(P W) = H(P^*W),
\end{align}
where the first inequality follows, since conditioning reduces the entropy.\\
The RV is defined in such a way that its conditional expectation can be upper bounded by the randomized IDF-capacity, similar to the procedure for deterministic feedback. Note here, that the output distribution $P_Y$ depends on the feedback probability $Q_f$, since the feedback strategy is a RV over the probability set $\mathcal{P}(\mathcal{F})$ for randomized encoding.\newline 
Now we define another RV
\begin{equation}
    U_t' = Y_t' - \mathbb{E}[Y_t'|Y_1,...,Y_{t-1}] \notag
\end{equation}
with the same properties as $U_t$ in \eqref{eq:def_ut}(a,b,c). \\
Similar to the proof for the deterministic encoder by using the Chebyshev inequality \eqref{eq:chebyshev} we get
\begin{equation}
    W^n(\mathcal{E}'|f) = \mathbb{P}\Big(-\log \prod\limits_{t=1}^{n} P_Y(Y_t|f) \leq nH(P^*W) + \alpha\sqrt{n}\Big) \geq 1-\nu. 
\end{equation}
We simply replace the set $\mathcal{E}$ by 
\begin{equation}
    \mathcal{E}' = \Big\{y^n| -\log W^n(y^n|f) \leq nH(P^* W) + \alpha\sqrt{n} =: \log K'  \Big\} \subset \mathcal{Y}^n \notag
\end{equation}
and Lemma \ref{lem:idf_subset_n} by the following Lemma \ref{lem:idf_subset_n_randomized}.

\begin{lemma}[{\cite[Lemma 4]{ahlswede_feedback}}]
\label{lem:idf_subset_n_randomized}
    For any randomized feedback strategy $f\in \mathcal{F}$, $f\sim Q_f$ and $\nu\in(0,1)$ with
    \begin{equation}
        \min\limits_{\mathcal{E}'\subset\mathcal{Y}^n:W^n(\mathcal{E}'|f)\geq 1-\nu} |\mathcal{E}'| \leq K' =: 2^{nH(P^*W)+\alpha\sqrt{n}}, 
    \end{equation}
    where $\alpha=\sqrt{\frac{\beta}{\nu}}$ and $\beta = |\mathcal{Y}|\log^23$, does for $N$ hold
    \begin{equation}
        N < 2^{n\log |\mathcal{Y}| K}.
    \end{equation}
\end{lemma}
\begin{proof}
    The proof is similar to Lemma \ref{lem:idf_subset_n} since we have already shown that $W^n(\mathcal{E}'|f)\geq 1-\nu$ holds. So the distinctness of $\mathcal{D}_i\cap\mathcal{E}_i'$ still holds (which can be proven in the same manner as for Lemma \ref{lem:idf_subset_n}) and we can use the same argument as in the proof for Lemma \ref{lem:idf_subset_n} to show that $N<2^{n\log|\mathcal{Y}|K}$ is satisfied. By replacing $H(W(\cdot|x^*))$ with $H(P^* W)$ we finally get
    \begin{equation}
        \limsup_{n \to \infty} \frac{1}{n} \log\log N < H(P^* W). 
    \end{equation}
\end{proof}

\begin{remark}
    The proof for the randomized encoder can be structured in the same five steps:
    \begin{enumerate}
        \item Introduction of the RV $Y_t^*$,
        \item Define the RV $U_t$ and estimate the PD of $\sum_{t=1}^{n}U_t$,
        \item Define subset $\mathcal{E}'\subset\mathcal{Y}^n$,
        \item Introduce Lemma \ref{lem:idf_subset_n_randomized},
        \item Coding rate.
    \end{enumerate}
\end{remark}
Therefore, we conclude that converse technique from Wolfowitz for transmission can be adapted, with modified parameters and settings, to prove the converse of the IDF.

\subsection{Comparison of Proofs: Wolfowitz's Proof versus Ahlswede and Dueck's Approach}\label{sec:comparison-tf-idf}
Finally, in this section we summarize the similarities and differences the converses for ID with feedback by Ahlswede and for transmission with feedback by Wolfowitz.

These converses share many similarities:

\begin{enumerate}
    \item introducing $Y_t^*$ and define $Y_t^*$ in such a way, that $\mathbb{E}[Y_t^*|Y_1,...,Y_{t-1}]$ can be upper bounded by the capacity,
    \item definition and properties of $U_t$,
    \item bounding the PD of $\sum_{t=1}^{n}U_t$ by $\nu$ with the Chebyshev inequality: \\$\mathbb{P}(\sum_{t=1}^{n}U_t\geq\alpha\sqrt{n})\leq\nu$,
    \item introduce subsets $\mathcal{E}$/ $\mathcal{C}(P_Y)$,
    \item bound the channel probability $W^n(\cdot|f)$ on the subsets $\mathcal{E}$/ $\mathcal{C}(P_Y)$ by $\nu$:\\ 
    Wolfowitz: $W^n(\mathcal{C}(P_Y)|f)\leq\nu$,\\
    Ahlswede: $W^n(\mathcal{E}|f)\geq1-\nu$,
    \item bound $W^n(\mathcal{D}_i\cap\mathcal{C}(P_Y)^c|f)$/ $W^n(\mathcal{D}_i\cap\mathcal{E}_i|f)$ with $\nu$ and the error of first kind $\varepsilon$:\\
    $W^n(\mathcal{D}_i\cap\mathcal{C}(P_Y)^c|f)>1-\varepsilon-\nu$/ $W^n(\mathcal{D}_i\cap\mathcal{E}_i|f)>1-\varepsilon-\nu$.
\end{enumerate}
There are also some notable differences:
\begin{enumerate}
    \item definition of $Y_t^*$, since $\mathbb{E}[Y_t^*|Y_1,...,Y_{t-1}]$ has to be bounded by the capacity and the capacities differ from each other,
    \item Ahlswede provides a concrete value for $\beta$ in the Chebyshev inequality,
    \item definition of the subsets $\mathcal{E}$/ $\mathcal{C}(P_Y)$: $Y_t^*$ is bounded by $\theta$/ $\log K$, such that \\
    \hspace{2cm} in transmission: $\theta \sim nC$ $\rightarrow$ exponential growth, \\
    \hspace{2cm} in identification: $K \sim 2^{nC}$ $\rightarrow$ double exponential growth, 
    \item definition of Lemma \ref{lem:wolfowitz}/ Lemma \ref{lem:idf_subset_n}: Ahlswede focuses on the cardinality of the subset $\mathcal{E}$,
    \item proof for Lemma \ref{lem:wolfowitz}/ Lemma \ref{lem:idf_subset_n}: \\
    Lemma \ref{lem:wolfowitz}: Wolfowitz uses definition of $\mathcal{C}(P_Y)$, \\
    Lemma \ref{lem:idf_subset_n}: Ahlswede is able to proof the distinctness of $\mathcal{D}_i\cap\mathcal{E}_i$ and uses this fact for the proof.
\end{enumerate}
\begin{remark}
    Ahlswede uses the distinctness of \( \mathcal{D}_i \cap \mathcal{E}_i \) to establish Lemma \ref{lem:idf_subset_n}. This property is demonstrated through the definitions of the first and second kind error probabilities. Since there exists only \underline{one} error probability in transmission codes, Wolfowitz cannot establish distinctness in this manner or any other. This underscores the fundamental difference between transmission and identification codes: the presence of two error probabilities in ID-codes necessitates different auxiliary statements (Lemma \ref{lem:wolfowitz} vs. Lemma \ref{lem:idf_subset_n}) and their respective proofs.
\end{remark}
\begin{remark}
The technique by Wolfowitz for proving the converse in the feedback case is suitable, since we consider the feedback by using conditional probabilistic measures, such like the conditional expectation. While information spectrum methods, used by Watanabe \cite{watanabe}, Hayashi \cite{hayashi} and Steinberg \cite{steinberg}, do not contain such a consideration, they are difficult to apply to this problem. \mmchange{For a further discussion of the applicability of the strong converse, especially in contexts where auxiliary random variables are needed, we refer the reader to \cite[Section VII]{rosenberger}.}
\end{remark}
In summary, Table \ref{tab:main_differences_feedback_wolfowitz_ahlswede} gives a really short overview of the main differences between Wolfowitz and Ahlswede. Note that the differences between Lemma \ref{lem:wolfowitz} and Lemma \ref{lem:idf_subset_n} are too extensive to write it down in a table. 
\begin{table}[h!]
    \centering
    \caption{Main differences between the converses for feedback by Wolfowitz and Ahlswede}
    \label{tab:main_differences_feedback_wolfowitz_ahlswede}
    \begin{tabular}{|c|c|c|}
        \hline
         & Wolfowitz & Ahlswede \\
         \hline
         code & transmission & identification \\
        $Y_t^*$ & $\log\frac{W(Y_t|f)}{P_Y(Y_t)}$ & $-\log W(Y_t|f)$ \\
        $\mathbb{E}[Y_t^*|Y_1,...,Y_{t-1}]$ & $\leq C$ & $\leq C_{IDF,d}$ \\
        subset & $\mathcal{C}(P_Y)$ & $\mathcal{E}$ \\
        threshold & $\theta\sim nC$ & $K \sim 2^{nC_{IDF,d}}$ \\
        C & $\sim 2^{nC}$ & $\sim2^{2^{nC_{IDF,d}}}$ \\
        \hline
        \end{tabular}
\end{table}
To visualize the direct comparison between the converse by Wolfowitz and by Ahlswede, Table \ref{tab:comparison_wolfowitz_ahlswede} shows the schematic steps compared side by side.

\begin{table}
\centering
\caption{\small Comparison of the converses for transmission with feedback by Wolfowitz \cite{wolfowitz} and for identification with feedback by Ahlswede \cite{ahlswede_feedback}}
\Rotatebox{90}{
    \begin{tabular}{|c|c|c|c|}
    \hline
         step& & \color{tu8}transmission \color{black}+ feedback (Wolfowitz \cite{wolfowitz}) & \color{tu2}identification \color{black}+ feedback (Ahlswede \cite{ahlswede_feedback}) \\
         \hline
         & scenario &  {\parbox[c]{9cm}{\vspace{0.2cm}\hspace{-0.72cm}\scalebox{0.92}{
\tikzstyle{farbverlauf} = [ top color=white, bottom color=white!80!gray]
\tikzstyle{block} = [draw,top color=white, bottom color=white!80!gray, rectangle, rounded corners,
minimum height=2em, minimum width=2.5em]

\tikzstyle{blocked} = [draw, fill=none, rectangle, rounded corners,
minimum height=4em, minimum width=.7
cm]
\tikzstyle{block1} = [draw, fill=none, rectangle, rounded corners,
minimum height=4em, minimum width=2
cm]
\tikzstyle{input} = [coordinate]
\tikzstyle{sum} = [draw, circle,inner sep=0pt, minimum size=2mm,  thick]

\scalebox{.93}{
\tikzstyle{arrow}=[draw,->] 
\begin{tikzpicture}[auto, node distance=2cm,>=latex']
\node[] (M) {${\color{blue}{m}}$};
\node[blocked, dashed, left=-0.6cm of M] (alice) {};
\node[block,right=.5cm of M] (enc) {Encoder};
\node[block, right=3cm of enc] (channel) {DMC};
\node[block, right=1cm of channel] (dec) {Decoder};
\node[blocked, dashed, right=0.5 cm of dec] (bob) {};
\node[align=left,right=.5cm of dec] (Mhat) {\color{blue}{$\hat{m}$}};
\node[input,right=.5cm of channel] (t1) {};
\node[input,below=1cm of t1] (t2) {};
\node[above= 0.3 cm of alice] (a) {Alice};
\node[above= 0.3 cm of bob] (b) {Bob};
\node[below=0.3 cm of alice] (message) {${\color{blue}{m}} \in \{1,\ldots,M\}$};
\draw[->,thick] (M) -- (enc);
\draw[->,thick] (enc) --node[above]{ $X_t={f}^t_{\color{blue}{m}}(Y^{t-1})$} (channel);
\draw[->,thick] (channel) --node[above]{$Y_t$} (dec);
\draw[->,thick] (dec) -- (Mhat);
\draw[-,thick] (t1) -- (t2);
\draw[->,thick] (t2) -| (enc);
\end{tikzpicture}}
}\vspace{0.2cm}}}& {\parbox[c]{9cm}{\vspace{0.2cm}\hspace{-0.63cm}\scalebox{0.85}{
\tikzstyle{farbverlauf} = [ top color=white, bottom color=white!80!gray]
\tikzstyle{block} = [draw,top color=white, bottom color=white!80!gray, rectangle, rounded corners,
minimum height=2em, minimum width=2.5em]

\tikzstyle{blocked} = [draw, fill=none, rectangle, rounded corners,
minimum height=4em, minimum width=.7
cm]
\tikzstyle{block1} = [draw, fill=none, rectangle, rounded corners,
minimum height=4em, minimum width=2
cm]
\tikzstyle{input} = [coordinate]
\tikzstyle{sum} = [draw, circle,inner sep=0pt, minimum size=2mm,  thick]

\scalebox{.93}{
\tikzstyle{arrow}=[draw,->] 
\begin{tikzpicture}[auto, node distance=2cm,>=latex']
\node[] (M) {${\color{blue}{i}}$};
\node[blocked, dashed, left=-0.6cm of M] (alice) {};
\node[block,right=.5cm of M] (enc) {Encoder};
\node[block, right=3cm of enc] (channel) {DMC};
\node[block, right=1cm of channel] (dec) {Decoder};
\node[block1, dashed, right=0.5 cm of dec] (bob) {};
\node[align=left,right=.5cm of dec] (Mhat) {Is ${i^\prime}$ sent? \\ Yes or No?};
\node[input,right=.5cm of channel] (t1) {};
\node[input,below=1cm of t1] (t2) {};
\node[above= 0.3 cm of alice] (a) {Alice};
\node[above= 0.3 cm of bob] (b) {Bob};
\node[below=0.3 cm of alice] (message) {${\color{blue}{i}} \in \{1,\ldots,N\}$};
\draw[->,thick] (M) -- (enc);
\draw[->,thick] (enc) --node[above]{ $X_t={f}^t_{\color{blue}{i}}(Y^{t-1})$} (channel);
\draw[->,thick] (channel) --node[above]{$Y_t$} (dec);
\draw[->,thick] (dec) -- (Mhat);
\draw[-,thick] (t1) -- (t2);
\draw[->,thick] (t2) -| (enc);
\end{tikzpicture}}
}\vspace{0.2cm}}}\\
         \hline
         & capacity & $C = \max\limits_{P} I(X;Y)$ (Shannon's transmission capacity) & $C_{\textrm{IDF,d}} = \max\limits_{x} H(W(\cdot|x)) = H(W(\cdot|x^*))$ \\
         \hline
         1.1.& define RV $Y_t^*$ & $Y_t^* = \log\frac{W(Y_t|f^t_{i}(Y^{t-1}))}{\color{black}P^*_Y(Y_t)}$ & $Y_t^* = \color{black}- \color{black}\log W(Y_t|f^t_{i}(Y^{t-1}))$ \\
         \hline
         1.2.& $\mathbb{E}[Y_t^*|Y^{t-1}]$ & \parbox{9cm}{$\mathbb{E}[Y_t^*|Y^{t-1}] \leq \color{black}C$ \\ $\rightarrow \mathbb{E}[Y_t^*|Y^{t-1}] = \sum\limits_{Y_t} Y_t^* \mathbb{P}(Y_t|Y^{t-1})$ \\ 
         \begin{addmargin}[12pt]{0pt}$= \sum\limits_{Y_t} \log\Big(\frac{W(Y_t|f^t_{i}(Y^{t-1}))}{P^*_Y(Y_t)}\Big) \cdot W(Y_t|f^t_{i}(Y^{t-1}))$ \\ $\leq \max_{P} I(X;Y) = C$, \end{addmargin}} & \parbox{9.8cm}{$\mathbb{E}[Y_t^*|Y_1,...,Y_{t-1}] \leq \color{black}C_{\textrm{IDF,d}}$ \\ $\rightarrow \mathbb{E}[Y_t^*|Y_1,...,Y_{t-1}] = \sum\limits_{Y_t} Y_t^* \mathbb{P}(Y_t|Y_1,...,Y_{t-1})$ \\ \begin{addmargin}[12pt]{0pt}$= - \sum\limits_{Y_t} \log W(Y_t|f^t_{i}(Y_1,...,Y_{t-1})) \cdot W(Y_t|f^t_{i}(Y_1,...,Y_{t-1}))$ \\ $= H(W(\cdot|x)) \leq H(W(\cdot|x^*)) = C_{\textrm{IDF,d}}$\end{addmargin}} \\
         \hline
         2.1.& define RV $U_t$ & \parbox{9.6cm}{$U_t = Y_t^* - \mathbb{E}[Y_t^*|Y_1,...,Y_{t-1}]$ \\$\text{with } \mathbb{E}[U_t|Y_1,...,Y_{t-1}] = 0, \mathbb{E}[U_t|U_s] = 0$, $s<t,$ \\ $\mathbb{E}[U_s U_t] = 0,$ $s\neq t$} & \parbox{9.8cm}{$U_t = Y_t^* - \mathbb{E}[Y_t^*|Y_1,...,Y_{t-1}]$ \\$\text{with } \mathbb{E}[U_t|Y_1,...,Y_{t-1}] = 0, \mathbb{E}[U_t|U_s] = 0$, $s<t,$\\ $\mathbb{E}[U_s U_t] = 0,$ $s\neq t$} \\
         \hline
         2.2.& PD of $\sum\limits_{t=1}^{n} U_t$ & \parbox{9.1cm}{let $\alpha > 0$ such that: \\ \begin{addmargin}[14pt]{0pt}$\sum\limits_{t=1}^{n} U_t \stackrel{2.1)}{=} \sum\limits_{t=1}^{n} Y_t^* - \sum\limits_{t=1}^{n} \mathbb{E}[Y_t^*|Y^{t-1}] \color{black}> \color{black}\alpha \sqrt{n}$ \\\end{addmargin} $\Leftrightarrow \sum\limits_{t=1}^{n} Y_t^* \color{black}> \color{black}\alpha\sqrt{n} + \sum\limits_{t=1}^{n}\mathbb{E}[Y_t^*|Y^{t-1}]$ \\with $\alpha\sqrt{n}+\sum\limits_{t=1}^{n} \underbrace{\mathbb{E}[Y_t^*|Y_1,...,Y_{t-1}]}_{\leq C}\leq \alpha\sqrt{n}+nC$ \\therefore: \\ $ \mathbb{P}(\sum\limits_{t=1}^{n} Y_t^* \color{black}>\color{black}\alpha\sqrt{n} + nC|Y_1,...,Y_{t-1})$ \\ $\color{black}\leq \color{black}\mathbb{P}(\sum\limits_{t=1}^{n} Y_t^* \color{black}>\color{black}\alpha\sqrt{n} + \sum\limits_{t=1}^{n}\mathbb{E}[Y_t^*|Y_1,...,Y_{t-1}]|Y_1,...,Y_{t-1})\\ = \mathbb{P}(\sum\limits_{t=1}^{n} U_t \color{black}> \color{black}\alpha\sqrt{n}|Y_1,...,Y_{t-1}) $} & \parbox{9.8cm}{let $\alpha > 0$ such that: \\ \begin{addmargin}[14pt]{0pt}$\sum\limits_{t=1}^{n} U_t \stackrel{2.1)}{=} \sum\limits_{t=1}^{n} Y_t^* - \sum\limits_{t=1}^{n} \mathbb{E}[Y_t^*|Y_1,...,Y_{t-1}] \color{black}\leq \color{black}\alpha \sqrt{n}$ \\ \end{addmargin}$\Leftrightarrow \sum\limits_{t=1}^{n} Y_t^* \color{black}\leq \color{black}\alpha\sqrt{n} + \sum\limits_{t=1}^{n}\underbrace{\mathbb{E}[Y_t^*|Y_1,...,Y_{t-1}]}_{\leq C_{\textrm{IDF;d}}} \leq \alpha\sqrt{n} + nC_{\textrm{IDF,d}} $ \\ therefore: \\$ \mathbb{P}(\sum\limits_{t=1}^{n} Y_t^* \color{black}\leq \color{black}\alpha\sqrt{n} + nC_{\textrm{IDF,d}}|Y_1,...,Y_{t-1}) $ \\ $\color{black}\geq \color{black}\mathbb{P}(\sum\limits_{t=1}^{n} Y_t^* \color{black}\leq\color{black} \alpha\sqrt{n} + \mathbb{E}[Y_t^*|Y_1,...,Y_{t-1}]|Y_1,...,Y_{t-1}) \\= \mathbb{P}(\sum\limits_{t=1}^{n} U_t \color{black}\leq \color{black}\alpha\sqrt{n}|Y_1,...,Y_{t-1})$ } \\
         \hline
    \end{tabular}
}

\label{tab:comparison_wolfowitz_ahlswede}
\end{table}

\begin{table}
    \centering
    \Rotatebox{90}{
    \begin{tabular}{|c|c|c|c|}
    \hline
    step& & \color{tu8}transmission\color{black} + feedback (Wolfowitz \cite{wolfowitz}) & \color{tu2}identification\color{black} + feedback (Ahlswede \cite{ahlswede_feedback}) \\
    \hline
    2.3.&\parbox{2cm}{Chebyshev \\ inequality \eqref{eq:chebyshev}} & \parbox{10.2cm}{for $X=\sum\limits_{t=1}^{n}U_t$, $\mu = \mathbb{E}[\sum\limits_{t=1}^{n}U_t] = 0$, $k=\frac{\alpha}{\sqrt{\beta}}=\frac{1}{\sqrt{\nu}}$ and $\beta$ as an upper bound for $\text{Var}[U_t]$, therefore \\$\sigma^2 = \text{Var}[\sum\limits_{t=1}^{n} U_t|Y_1,...,Y_{t-1}] =\sum\limits_{t=1}^{n} \text{Var}[U_t|Y_1,...,Y_{t-1}] \leq n\beta $: \\ $\mathbb{P}(|X-\mu|\geq k\sigma) = \mathbb{P}(\sum\limits_{t=1}^{n} U_t \geq \alpha\sqrt{n}|Y_1,...,Y_{t-1}) \leq \frac{1}{k^2} = \nu$ } & \parbox{10cm}{for $X=\sum\limits_{t=1}^{n}U_t$, $\mu = \mathbb{E}[\sum\limits_{t=1}^{n}U_t|Y_1,...,Y_{t-1}] = 0$, $k=\frac{\alpha}{\sqrt{\beta}} = \frac{1}{\sqrt{\nu}}$ and $\color{black}\beta := |\mathcal{Y}|\log^2 3\color{black}$ as an upper bound for $\text{Var}[U_t]$, therefore \\$\sigma^2 = \text{Var}[\sum\limits_{t=1}^{n} U_t|Y_1,...,Y_{t-1}] =\sum\limits_{t=1}^{n} \text{Var}[U_t|Y_1,...,Y_{t-1}] \leq n\beta $: \\ $\mathbb{P}(|X-\mu|\geq k\sigma) = \mathbb{P}(\sum\limits_{t=1}^{n} U_t \geq \alpha\sqrt{n}|Y_1,...,Y_{t-1}) \leq \frac{1}{k^2} = \nu$ }  \\
    \hline
    3.&\parbox{2cm}{define \\subset} & \parbox{10.2cm}{$\mathcal{C}(P_Y) =\{y^n| \log\frac{W^n(y^n|f)}{P_Y(y^n)} \color{black}> \color{black} n\color{black}C \color{black}+ \alpha\sqrt{n} =: \color{black}\theta \color{black}\}$ \\ note that \color{tu110}$W^n(\mathcal{C}(P^*_Y)|f)\color{black}=\mathbb{P}(\sum\limits_{t=1}^{n} Y_t^* >\alpha\sqrt{n} + nC|U_1,...,U_{t-1})$\\ \begin{addmargin}[100pt]{0pt}\color{black}$\overset{\text{step}}{\underset{\text{\tiny{2.2)+2.3)}}}{\leq}}\color{tu110}\nu$\color{black}\end{addmargin}} & \parbox{10.2cm}{$\mathcal{E} = \{ y^n | -\log W^n(y^n|f) \color{black}\leq \color{black}n\color{black}C_{\textrm{IDF,d}} \color{black}+ \alpha\sqrt{n} =:\color{black} \log K \color{black}\}$\\ note that \color{tu110}$W^n(\mathcal{E}|f)\color{black}=\mathbb{P}(\sum\limits_{t=1}^{n} Y_t^* \leq\alpha\sqrt{n} + nC_{\textrm{IDF,d}}|Y_1,...,Y_{t-1} )$\\ \begin{addmargin}[16pt]{0pt} \color{black}$= 1- \mathbb{P}(\sum\limits_{t=1}^{n} Y_t^* \geq\alpha\sqrt{n} + nC_{\textrm{IDF,d}}|Y_1,...,Y_{t-1} )\overset{\text{step}}{\underset{\text{\tiny{2.2)+2.3)}}}{\geq}} \color{tu110}1-\nu$\color{black}\end{addmargin} } \\
    \hline
    4.1.&\parbox{2.05cm}{\color{tu8}Lemma \ref{lem:wolfowitz}\color{black}/ \\ \color{tu2}Lemma \ref{lem:idf_subset_n}}\color{black} & \parbox{10.2cm}{for  
    $\color{tu110}\min\limits_{P_Y}\max\limits_{f} W^n(\mathcal{C}(P_Y)|f) < \nu $ \color{black}$\Rightarrow M < \frac{2^\theta}{1-\nu-\varepsilon}$ \\ \textit{Proof: } \\ $2^\theta\cdot P_{Y}^*(\mathcal{D}_i\cap \mathcal{C}(P_Y^*)) \color{black}\overset{\text{def. of}}{\underset{\mathcal{C}(P_Y)}{\geq}} \color{black}W^n(\mathcal{D}_i\cap \mathcal{C}(P_Y^*)^c|f) \overset{\text{step}}{\underset{\text{4.2)}}{>}} 1-\nu-\varepsilon$ \\ $\Leftrightarrow 2^\theta \geq 2^\theta\underbrace{\sum\limits_{i=1}^{M} P_{Y}^*(\mathcal{D}_i\cap \mathcal{C}(P_Y^*))}_{\leq1} > \sum\limits_{i=1}^{M} 1-\nu-\varepsilon = M(1-\nu-\varepsilon)$ } & \parbox{9.8cm}{$\min\limits_{\mathcal{E}\subset\mathcal{Y}^n:\color{tu110}W^n(\mathcal{E}|f)\geq 1-\nu} \color{black}|\mathcal{E}| \leq K =: 2^{nC_{IDF,d}+\alpha\sqrt{n}}$  $\Rightarrow N<2^{n\log|\mathcal{Y}|K}$ \\ \textit{Proof: } \\ $N\color{black}\overset{\mathcal{D}_i\cap\mathcal{E}_i}{\underset{\text{distinct}}{\leq}} \color{black}\sum\limits_{i=0}^{K} {{|\mathcal{Y}|^n}\choose{i}} {\leq} (|\mathcal{Y}|^n)^K = (2^{n\log|\mathcal{Y}|^n})^K = 2^{n\log|\mathcal{Y}|K} $  } \\
    \hline
    4.2.&\parbox{2cm}{bound \\ \footnotesize{$W^n(\mathcal{D}_i\cap\cdot|f_i)$}} & \parbox{10.3cm}{$ W^n(\mathcal{D}_i\cap\mathcal{C}(P_{Y^*})^{\color{black}c}|f_i) {\geq} \underbrace{W^n(\mathcal{D}_i|f_i)}_{>1-\varepsilon} - \underbrace{W^n(\mathcal{C}(P_{Y^*})|f_i)}_{\color{tu110}<\nu \text{\tiny{ by step 3)}}}$ \\ \begin{addmargin}[100pt]{0pt} $> 1-\varepsilon-\nu $, with $\varepsilon$ as type $I$ error \end{addmargin}} & \parbox{9.8cm}{$W^n(\mathcal{D}_i\cap\mathcal{E}_i|f_i) {\geq} \underbrace{W^n(\mathcal{E}_i|f_i)}_{\color{tu110}\geq1-\nu\text{\tiny{ by step 3)}}} - \underbrace{W^n(\mathcal{D}_i^c|f_i)}_{<\varepsilon} > 1- \nu - \varepsilon $, \\ with \color{black} $\mathcal{D}_i\cap\mathcal{E}_i$ distinct \color{black} and $\varepsilon$ as type $I$ error} \\
    \hline
    5.&\parbox{2cm}{coding rate} & \parbox{10.2cm}{$M<\frac{2^\theta}{1-\nu-\varepsilon}$ \\ $ \Leftrightarrow \log M < \theta - \log(1-\nu-\varepsilon) = nC+\alpha\sqrt{n}-\log(1-\nu-\varepsilon)$ \\ $\Leftrightarrow \frac{1}{n}\log M< C + \underbrace{\frac{\alpha}{\sqrt{n}}}_{\xrightarrow[n\to\infty]{} 0} + \underbrace{\frac{1}{n}\log(1-\nu-\varepsilon)}_{\xrightarrow[n\to\infty]{} 0}$ \\ $\Rightarrow \limsup_{n \to \infty}\frac{1}{n}\log M < C$ } & \parbox{10cm}{$N<2^{n\log|\mathcal{Y}|K} \Leftrightarrow \log N < n\log|\mathcal{Y}|2^{nC_{IDF,d}+\alpha\sqrt{n}} $ \\ $\Leftrightarrow \log\log N < \log(n\log|\mathcal{Y}|) + nC_{IDF,d}+\alpha\sqrt{n}$ \\ $\Leftrightarrow \frac{1}{n} \log\log N < \underbrace{\frac{1}{n}\log(n\log|\mathcal{Y}|)}_{\xrightarrow[n\to\infty]{} 0} + C_{IDF,d} + \underbrace{\frac{\alpha}{\sqrt{n}}}_{\xrightarrow[n\to\infty]{} 0}$ \\ $\Rightarrow \limsup_{n \to \infty}\frac{1}{n}\log\color{black}\log \color{black}N < C_{IDF,d}$ } \\
    \hline
    \end{tabular}
    }
\end{table}

\newpage

\section*{Acknowledgments}
C. Deppe, W. Labidi, and Y. Zhao acknowledge the financial support by the Federal Ministry of Education and Research
of Germany (BMBF) in the program of “Souverän. Digital. Vernetzt.”. Joint project 6G-life, project identification number: 16KISK002.
E. Jorswieck and M. Mross are supported by
the Federal Ministry of Education and Research (BMBF, Germany) through
the Program of “Souveran. Digital. Vernetzt.” Joint Project 6G-Research and
Innovation Cluster (6G-RIC) under Grant 16KISK031
W. Labidi were further supported in part by the BMBF within the national initiative on Post Shannon Communication (NewCom) under Grant 16KIS1003K. C.\ Deppe was further supported in part by the BMBF within NewCom under Grant 16KIS1005. C. Deppe was also supported by the DFG within the project DE1915/2-1. 

\bibliographystyle{splncs04}
\bibliography{references} 
\end{document}